\ifx\SoCG\undefined%
   \documentclass[12pt]{article}%
   \def\UseBibLatex{1}%
   \providecommand{\SoCGVer}[1]{}%
   \providecommand{\NotSoCGVer}[1]{#1}%

\else%
\makeatletter
\def\input@path{{lipics/}{../lipics/}}
\makeatother
\documentclass[a4paper,USenglish,cleveref,autoref,thm-restate]%
{socg-lipics-v2021}
\def\ProofInAppendix{1}

\providecommand{\SoCGVer}[1]{#1}%
\providecommand{\NotSoCGVer}[1]{}%
\fi

\ifx\sarielComp\undefined%
\newcommand{\SarielComp}[1]{}
\newcommand{\NotSarielComp}[1]{#1}%
\else
\newcommand{\SarielComp}[1]{#1}%
\newcommand{\NotSarielComp}[1]{}%
\fi
\newcommand{\IfPrinterVer}[2]{#2}%

\usepackage{mathtools }

\NotSoCGVer{%
\usepackage[in]{fullpage}%
}

\usepackage{amsmath}%
\usepackage{amssymb}%
\usepackage[cmyk]{xcolor}%

\NotSoCGVer{%
   \usepackage{euscript}%
}
\NotSoCGVer{%
   \usepackage[amsmath,thmmarks]{ntheorem}%
   \theoremseparator{.}%
}

\NotSoCGVer{%
   \usepackage{titlesec}%
   \titlelabel{\thetitle. }%

   \titleformat{\paragraph}[runin]
   {\normalfont\bfseries}
   {\theparagraph}
   {1em}
   {\addperiod}

   \newcommand{\addperiod}[1]{#1.}%
}

\usepackage{graphicx}%
\usepackage{xcolor}%
\usepackage{hyperref}%
\usepackage{bm}

\usepackage{caption}%

\SarielComp{\usepackage{sariel_colors}}%

\IfPrinterVer{%
   \usepackage{hyperref}%
}{%
   \usepackage{hyperref}%
   \hypersetup{%
      breaklinks,%
      colorlinks=true,%
      urlcolor=[rgb]{0.25,0.0,0.0},%
      linkcolor=[rgb]{0.5,0.0,0.0},%
      citecolor=[rgb]{0,0.2,0.445},%
      filecolor=[rgb]{0,0,0.4},
      anchorcolor=[rgb]={0.0,0.1,0.2}%
   }
}

\providecommand{\BibLatexMode}[1]{}
\providecommand{\BibTexMode}[1]{#1}

\ifx\UseBibLatex\undefined%
    \renewcommand{\BibLatexMode}[1]{}
    \renewcommand{\BibTexMode}[1]{#1}
\else
    \renewcommand{\BibLatexMode}[1]{#1}
    \renewcommand{\BibTexMode}[1]{}
\fi

\newcommand{\UsePackage}[1]{%
  \IfFileExists{../styles/#1.sty}{%
      \usepackage{../styles/#1}%
   }{%
      \IfFileExists{./styles/#1.sty}{%
         \usepackage{styles/#1}%
      }{%
         \usepackage{#1}%
      }%
   }%
}

\UsePackage{proofapnd}

\BibLatexMode{%
   \usepackage[bibencoding=utf8,style=alphabetic,%
   backend=biber,sortlocale=en_US]{biblatex}%
   \UsePackage{sariel_biblatex}%
   \usepackage[english]{babel}%
   \usepackage{csquotes}
   \renewcommand*{\multicitedelim}{\addcomma\space}%
}

\SoCGVer{%
   \theoremstyle{plain}%

   \newtheorem{fact}[theorem]{Fact}
   \newtheorem{invariant}[theorem]{Invariant}
   \newtheorem{question}[theorem]{Question}
   \newtheorem{prop}[theorem]{Proposition}
   \newtheorem{openproblem}[theorem]{Open Problem}

   \theoremstyle{plain}%
   \newtheorem{defn}[theorem]{Definition}
   \newtheorem{problem}[theorem]{Problem}
   \newtheorem{xca}[theorem]{Exercise}
   \newtheorem{exercise_h}[theorem]{Exercise}
   \newtheorem{assumption}[theorem]{Assumption}%

   \newtheorem{proofof}{Proof of\!}%
}%

\NotSoCGVer{%
\theoremseparator{.}%

\theoremstyle{plain}%
\newtheorem{theorem}{Theorem}[section]

\newtheorem{lemma}[theorem]{Lemma}

\newtheorem{corollary}[theorem]{Corollary}

\newtheorem{observation}[theorem]{Observation}

\theoremstyle{plain}%
\theoremheaderfont{\sf} \theorembodyfont{\upshape}%
\newtheorem*{remark:unnumbered}[theorem]{Remark}%
\newtheorem{remark}[theorem]{Remark}%

\newtheorem{defn}[theorem]{Definition}
\newtheorem{example}[theorem]{Example}

\newcommand{\myqedsymbol}{\rule{2mm}{2mm}}

\theoremheaderfont{\em}%
\theorembodyfont{\upshape}%
\theoremstyle{nonumberplain}%
\theoremseparator{}%
\theoremsymbol{\myqedsymbol}%
\newtheorem{proof}{Proof:}%

}

\definecolor{nalmostblack}{rgb}{0, 0, 0.7}
\providecommand{\emphic}[2]{%
   \textcolor{nalmostblack}{%
      \textbf{\emph{#1}}}%
   \index{#2}}

\providecommand{\emphi}[1]{\emphic{#1}{#1}}
\providecommand{\emphcolor}[1]{\textcolor{nalmostblack}{{#1}}}

\definecolor{almostblack}{rgb}{0, 0, 0.5}
\providecommand{\emphw}[1]{{\emph{{\textcolor{almostblack}{#1}}}}}%

\numberwithin{figure}{section}%
\numberwithin{table}{section}%
\numberwithin{equation}{section}%

\newcommand{\SarielThanks}[1]{\thanks{Department of Computer Science;
      University of Illinois; 201 N. Goodwin Avenue; Urbana, IL,
      61801, USA; %
      \href{mailto:sariel.spam@illinois.edu}%
      {sariel@illinois.edu}; %
      \url{http://sarielhp.org/}. #1}}

\newcommand{\HLink}[2]{{\hyperref[#2]{#1~\ref*{#2}}}}
\newcommand{\HLinkS}[2]{{\hyperref[#2]{#1\ref*{#2}}}}
\newcommand{\HLinkY}[2]{\hyperref[#2]{#1}}
\newcommand{\HLinkSuffix}[3]{\hyperref[#2]{#1\ref*{#2}{#3}}}

\newcommand{\figlab}[1]{\label{fig:#1}}
\newcommand{\figref}[1]{\HLink{Figure}{fig:#1}}

\newcommand{\thmlab}[1]{{\label{theo:#1}}}
\newcommand{\thmref}[1]{\HLink{Theorem}{theo:#1}}

\newcommand{\thmrefY}[2]{\HLinkY{#2}{theo:#1}}

\newcommand{\corlab}[1]{\label{cor:#1}}

\newcommand{\seclab}[1]{\label{sec:#1}}
\newcommand{\secref}[1]{\HLink{Section}{sec:#1}}

\newcommand{\Rects}{\Mh{\mathcal{R}}}%
\newcommand{\Objs}{\Mh{\mathcal{O}}}%

\newcommand{\RectsInfX}[1]{\Rects^{\rectangleC}\pth{#1}}%

\newcommand{\apndlab}[1]{\label{apnd:#1}}
\newcommand{\apndref}[1]{\HLink{Appendix}{apnd:#1}}

\newcommand{\obslab}[1]{\label{observation:#1}}
\newcommand{\obsref}[1]{\HLink{Observation}{observation:#1}}

\newcommand{\lemlab}[1]{\label{lemma:#1}}
\newcommand{\lemref}[1]{\HLink{Lemma}{lemma:#1}}%
\newcommand{\Xlemref}[1]{\noexpand{\noexpand\HLink{Lemma}{lemma:#1}}}%
\newcommand{\Xcorref}[1]{\noexpand{\noexpand\HLink{Corollary}{cor:#1}}}%

\providecommand{\eqlab}[1]{}%
\renewcommand{\eqlab}[1]{\label{equation:#1}}
\newcommand{\Eqref}[1]{\HLinkSuffix{Eq.~(}{equation:#1}{)}}

\providecommand{\remove}[1]{}%
\newcommand{\Set}[2]{\left\{ #1 \;\middle\vert\; #2 \right\}}
\newcommand{\pbrcx}[1]{\left[ {#1} \right]}%
\newcommand{\Prob}[1]{\mathop{\mathbf{Pr}}\!\pbrcx{#1}}
\newcommand{\Ex}[2][\!]{\mathop{\mathbf{E}}#1\pbrcx{#2}}

\newcommand{\ceil}[1]{\left\lceil {#1} \right\rceil}
\newcommand{\floor}[1]{\left\lfloor {#1} \right\rfloor}

\newcommand{\cardin}[1]{\left| {#1} \right|}%

\renewcommand{\th}{th\xspace}

\renewcommand{\Re}{\mathbb{R}}%

\usepackage[inline]{enumitem}

\newlist{compactenumA}{enumerate}{5}%
\setlist[compactenumA]{topsep=0pt,itemsep=-1ex,partopsep=1ex,parsep=1ex,%
   label=(\Alph*)}%

\newlist{compactenuma}{enumerate}{5}%
\setlist[compactenuma]{topsep=0pt,itemsep=-1ex,partopsep=1ex,parsep=1ex,%
   label=(\alph*)}%

\newlist{compactenumI}{enumerate}{5}%
\setlist[compactenumI]{topsep=0pt,itemsep=-1ex,partopsep=1ex,parsep=1ex,%
   label=(\Roman*)}%

\newlist{compactenumi}{enumerate}{5}%
\setlist[compactenumi]{topsep=0pt,itemsep=-1ex,partopsep=1ex,parsep=1ex,%
   label=(\roman*)}%

\newlist{compactitem}{itemize}{5}%
\setlist[compactitem]{label=\ensuremath{\bullet}}%
\setlist[compactitem]{topsep=0pt,itemsep=-1ex,partopsep=1ex,parsep=1ex,%
   label=\ensuremath{\bullet}}%

\usepackage{stmaryrd}%
\providecommand{\IntRange}[1]{\mleft\llbracket #1 \mright\rrbracket}
\newcommand{\IRX}[1]{\IntRange{#1}}%

\usepackage{wasysym}

\newcommand{\si}[1]{#1}

\providecommand{\Mh}[1]{#1}%

\newcommand{\eps}{\varepsilon}

\newcommand{\etal}{\textit{et~al.}\xspace}

\newcommand{\Term}[1]{\textsf{#1}}

\newcommand{\StavThanks}[1]{%
   \thanks{Department of Computer Science; University of Illinois; 201
      N. Goodwin Avenue; Urbana, IL, 61801, USA; %
      \href{mailto:stava2.spam@illinois.edu}%
      {stava2@illinois.edu}; %
      \url{https://publish.illinois.edu/stav-ashur}. %
      #1}}

\providecommand{\G}{\Mh{G}}%
\renewcommand{\G}{\Mh{G}}%

\newcommand{\VV}{\Mh{V}}
\newcommand{\EE}{\Mh{E}}
\newcommand{\EGX}[1]{\EE\pth{#1}}%

\newcommand{\PSA}{\Mh{Q}}%

\newcommand{\pa}{\Mh{u}}
\newcommand{\pb}{\Mh{v}}

\newcommand{\PA}{\Mh{U}}%
\newcommand{\PB}{\Mh{V}}%

\newcommand{\Y}{\Mh{Y}}

\newcommand{\heightX}[1]{\mathrm{height}\pth{#1}}

\newcommand{\Line}{\Mh{\ell}}%

\newcommand{\Tree}{\Mh{T}}%

\newcommand{\region}{\Mh{\mathcalb{r}}}%

\usepackage{stmaryrd,graphicx}



\newcommand{\ts}{\hspace{0.6pt}}
\newcommand{\radom}{{\ts\searrow}}%

\newcommand{\rectangleC}{\vcenter{\hbox{\scalebox{1.0}[0.5]{$\square$}}}}
\newcommand{\rankX}[1]{\mathrm{rank}\pth{#1}}

\newcommand{\Forest}{\Mh{\mathcal{F}}}

\newcommand{\rectsymbol}{\vcenter{\hbox{\scalebox{1.0}[0.5]{$\square$}}}}

\newcommand{\rectY}[2]{{\rectsymbol\pth{#1,#2}}}

\renewcommand{\S}{\mathbb{S}}%

\newcommand{\BC}{\Mh{\mathcal{C}}}%
\newcommand{\BCA}{\Mh{\mathcal{D}}}%
\newcommand{\BCB}{\Mh{\mathcal{F}}}%

\newcommand{\B}{\Mh{{B}}}%

\newcommand{\WeightX}[1]{\Mh{\omega} \pth{#1}}

\newcommand{\SaveContent}[2]{%
   \expandafter\newcommand{#1}{#2}%
}

\DeclareFontFamily{U}{BOONDOX-calo}{\skewchar\font=45 }
\DeclareFontShape{U}{BOONDOX-calo}{m}{n}{<-> s*[1.05] BOONDOX-r-calo}{}
\DeclareFontShape{U}{BOONDOX-calo}{b}{n}{<-> s*[1.05] BOONDOX-b-calo}{}
\DeclareMathAlphabet{\mathcalb}{U}{BOONDOX-calo}{m}{n}
\SetMathAlphabet{\mathcalb}{bold}{U}{BOONDOX-calo}{b}{n}
\DeclareMathAlphabet{\mathbcalb}{U}{BOONDOX-calo}{b}{n}

\newcommand{\maximaX}[1]{\mathsf{max}_\domby\pth{#1}}%
\newcommand{\amaximaX}[1]{\mathsf{max}_\adom\pth{#1}}%
\newcommand{\domMN}[1]{\mathsf{min}_\domby\pth{#1}}%
\newcommand{\adomMN}[1]{\mathsf{min}_\adom\pth{#1}}%

\newcommand{\BHX}[1]{\mathsf{bh}\pth{#1}}%
\newcommand{\argmax}{\mathsf{argmax}}%

\newcommand{\Arr}{\Mh{\mathcal{A}}}
\newcommand{\ArrX}[1]{\Arr\pth{#1}}

\newcommand{\BH}{\mathcal{B}}%

 \smallskip%

\SoCGVer{%
   \newcommand{\myparagraph}[1]{%
      \bigskip%
      \noindent%
      \textbf{#1.}
   }%
}
\NotSoCGVer{%
   \newcommand{\myparagraph}[1]{%
      \paragraph{#1}
   }%
}

\providecommand{\TPDF}[2]{\texorpdfstring{#1}{#2}}

\newcommand{\edgesR}{E^{\rectangleC}}%

\newcommand{\RIG}{\Term{RIG}\xspace}%

\newcommand{\RIGraph}{\Mh{\mathcal{G}}}
\newcommand{\RIGX}[1]{\RIGraph^{\rectangleC}\pth{#1}}%
\newcommand{\RIGDX}[1]{\RIGraph^{\rectangleC}_\domby\pth{#1}}%
\newcommand{\RIGAX}[1]{\RIGraph^{\rectangleC}_\adom\pth{#1}}%
\newcommand{\RIGKY}[2]{\RIGraph^{\rectangleC}_{\leq #1}\pth{#2}}%
\newcommand{\rA}{\mathsf{R}}

\DefineNamedColor{named}{AlgorithmColor}{cmyk}{0.07,0.90,0,0.34}

\newcommand{\AlgorithmI}[1]{{%
      \textcolor[named]{AlgorithmColor}{\texttt{\bf{#1}}}%
   }}
\newcommand{\listX}[1]{\Mh{L}\pth{#1}}%
\newcommand{\listSX}[1]{\Mh{L}_{#1}}%

\newcommand{\pushOp}{\AlgorithmI{push}\xspace}%
\newcommand{\popOp}{\AlgorithmI{pop}\xspace}%
\newcommand{\reportOp}{\AlgorithmI{report}\xspace}%

\usepackage{xspace}%
\usepackage{mleftright}%

\providecommand{\Mh}[1]{#1}

\newcommand{\p}{\Mh{p}}%
\newcommand{\q}{\Mh{q}}%
\newcommand{\pc}{\Mh{u}}%
\newcommand{\pd}{\Mh{v}}%

\newcommand{\concat}{|}

\newcommand{\Quad}{\Mh{\mathcal{Q}}}%
\newcommand{\QuadX}[1]{\Mh{\mathcal{Q}}\pth{#1}}%

\newcommand{\rect}{{\Mh{\mathsf{r}}}}%
\newcommand{\rectA}{\Mh{\mathsf{b}}}%
\newcommand{\rectB}{\Mh{\mathsf{d}}}%

\newcommand{\pl}{\Mh{\mathcalb{l}}}%
\newcommand{\pr}{\Mh{\mathcalb{r}}}%

\providecommand{\L}{\Mh{L}}%
\renewcommand{\L}{\Mh{L}}%
\newcommand{\R}{\Mh{{R}}}%

\newcommand{\BSC}{\Mh{\mathsf{B}}}%
\newcommand{\TSC}{\Mh{\mathsf{T}}}%

\newcommand{\interX}[1]{\mathrm{int}\pth{#1}}%

\newcommand{\ShadowC}{\Mh{\mathcal{S}}}%
\newcommand{\shadowY}[2]{\mathcal{S}\pth{#1,#2}}%

\newcommand{\domby}{{\nearrow}}%
\newcommand{\adom}{{\ts\nwarrow}}%
\newcommand{\Ndomby}{{\swarrow}}%
\newcommand{\Nadom}{{\ts\searrow}}%

\newcommand{\pth}[2][\!]{\mleft({#2}\mright)}%

\renewcommand{\P}{\Mh{\ensuremath{P}}\xspace}%

\BibLatexMode{\bibliography{rect_delaunay}}

\begin{document}

\title{On the Rectangles Induced by Points}

\SoCGVer{%
   \author{Anonymous}{}{}{}{}{}%
}%
\NotSoCGVer{%
   \author{%
      Stav Ashur%
      \StavThanks{}%
      \and%
      Sariel Har-Peled%
      \SarielThanks{Work on this paper was partially supported by a
         NSF AF award CCF-2317241.  }%
   }%
}

\SoCGVer{%
   \remove{%
      \author{Stav Ashur}%
      {Department of Computer Science, University of Illinois, 201
         N. Goodwin Avenue, Urbana, IL 61801, USA}%
      {stava2@illinois.edu}%
      {https://orcid.org/0000-0003-0533-8978}%
      {}%
      \author{Sariel Har-Peled}%
      {Department of Computer Science, University of Illinois, 201
         N. Goodwin Avenue, Urbana, IL 61801, USA}%
      {sariel@illinois.edu}%
      {https://orcid.org/0000-0003-2638-9635}%
      {Work on this paper was partially supported by a NSF AF award
         CCF-1907400.}%
   }%
}%

\SoCGVer{%
   \authorrunning{Anony Mous and Friends}%
}

\SoCGVer{%
   \ccsdesc[500]{Theory of computation~Computational geometry}%
   \keywords{Geometric graphs, Fault-tolerant spanners}%
}

\NotSoCGVer{\date{\today}}

\maketitle

\begin{abstract}
    A set $\P$ of $n$ points in the plane induces a set of
    Delaunay-type axis-parallel rectangles $\Rects$, potentially of
    quadratic size, where a rectangle is in $\Rects$ if it has two
    points of $\P$ as corners, and no other point of $\P$ in it. We
    study various algorithmic problems related to this set of
    rectangles, including how to compute it in near linear time, and
    handle various algorithmic tasks on it, such as computing its
    union and depth.  The set of rectangles $\Rects$ induces the
    \emphw{rectangle influence graph} $\G = (\P,\Rects)$, which we
    also study. Potentially our most interesting result is showing
    that this graph can be described as the union of $O(n)$ bicliques,
    where the total weight of the bicliques is $O(n \log^2 n)$, where
    the weight of a bicliques is the number of its vertices.
\end{abstract}

\section{Introduction}

For two points $\p, \q$ in the plane, let $\rectY{\p}{\q}$ denote the
axis-aligned rectangle having $\p$ and $\q$ as corners.  For a set
$\P$ of $n$ points in the plane (in general position), consider the
set of all such (closed) rectangles that contain only their two
defining corners:
\begin{equation}
    \Rects
    =
    \RectsInfX{\P}
    =%
    \Set{\rectY{\p}{\q}}{ \p, \q \in \P \text{ and }
       \cardin{\rectY{\p}{\q}\cap \P} = 2   }.
    \eqlab{rect:set}
\end{equation}

The graph induced by this set of rectangles is the \emphi{rectangle
   influence graph} (\emphcolor{\RIG}) of $\P$, denoted by
$\RIGraph = \RIGX{\P} = (\P, \edgesR)$, with
\begin{math}
    \edgesR = \Set{ \p \q}{\p, \q \in \P, \rectY{\p}{\q} \in
       \RectsInfX{\P}}.
\end{math}
The (implicit) graph $\RIGraph$ (and the associated set $\Rects$) IS
quite interesting, as IT can have a quadratic number of edges, but ITS
fully defined by the $n$ points of $P$.

\myparagraph{Previous work on rectangle influence graphs} %
\RIG{}s were first defined by Ichino and Sklansky \cite{is-rngmfv-85}
who studied it as a special case of relative neighborhood graphs
\cite{t-rngfps-80}. Several papers were dedicated to the problem of
reporting rectangularly visible pairs $(\p,\q)$, i.e. whose rectangle
of influence $\rectY{\p}{\q}$ is empty, and constructing data
structures capable of reporting all points rectangularly visible from
a query point \cite{mow-vv-87, ow-rv-88, gno-fadedd-89}. These
settings were later generalized by de Berg \etal \cite{dbco-gadp-92}
to feature an input consisting of a set of points and a set of
visibility obstacles. A closely related and well-studied problem is
rectangle influence drawability \cite{egllmw-rridg-94, llmw-ridp-98,
   bbm-ridgwf3c-99, mn-ridfcpg-05, mmn-oridi-09,zv-oridpg-09,
   \si{sz-cridit-10}} in which one needs to determine whether a given
graph $\G=(V,E)$ has a realization of its vertices as a set $\P$ of
points in the plane, such that $G = \RIGX{\P}$ (or $\G$ is a subgraph
of $\RIGraph$ in the weak rectangle influence drawability problem).

\begin{figure}[h]
    \centering
    \begin{tabular}{c|c|c}
      \includegraphics[page=1,scale=0.99]{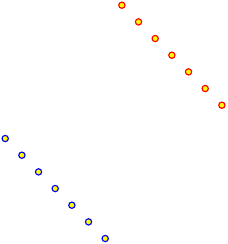}%
      \qquad%
      \phantom{}
      &%
        \includegraphics[page=2,scale=0.99]{figs/two_diagonals}%
      &%
        \qquad%
        \includegraphics[page=3,scale=0.99]{figs/two_diagonals}%
      \\
      (A)&(B)&(C)
    \end{tabular}
      \hfill%
      \phantom{}%
      \caption{(A) The point set. (B) The resulting \RIG.  (C) The
         union of all the associated rectangles, forming the box hull
         of the point set.  This is an example of the \RIG with
         maximum number of edges, and it is formed by two diagonals
         with the same number of points, where points in the top
         diagonal dominates all the points in the lower diagonal. }
      \figlab{two:diags}
\end{figure}

Alon \etal \cite{afk-sppsb-85} proved upper bounds on pairs of points
separated by standard $d$-dimensional boxes, or in the $2$-dimensional
case, pairs of points whose rectangle of influence is empty, and gave
an exact bound of $n^2/4+n-2$ on the maximum number of edges of
$\RIGX{\P}$ in the two dimensional case, see \figref{two:diags}.  Chen
\etal \cite{cpst-dgpsp-09} studied the independence number of
$\RIGX{\P}$, and showed that if $\P$ is a set of uniformly distributed
random points the independence number is sub-linear with high
probability for large enough values of $n$. They also showed that the
expected number of edges in these settings is $\Theta(n\log n)$. As
for the graph's chromatic number, $O(\sqrt{n})$ was shown by Har-Peled
and Smorodinsky \cite{hs-cfcps-05}. This was improved to
$O(n^{0.382})$ by Ajwani \etal \cite{aegr-cfcrr-12}, and later to
$O(n^{0.368})$ by Chan \cite{c-cfcpr-12}. Getting a better bound on
the chromatic number of this graph is still open.

Aronov \etal \cite{adh-wrg-14} studied \emph{witness rectangle
   graphs}, a closely related notion in which an edge $\p\q$ exists if
and only if the rectangle of influence $\rectY{\p}{\q}$ contains a
member of a given witness set. This is the positive variant of witness
rectangle graphs, where the negative variant is similar to the
generalization of \RIG{}s studied in \cite{dbco-gadp-92}.

\myparagraph{Our take}
We take a somewhat different approach than the above works. We are
interested in the (implicitly defined) set of rectangles $\Rects$. We
are interested in how to compute this set of rectangles efficiently,
how to represent it, and how to do various algorithmic tasks on this
set of rectangles. In particular, the union of the rectangles of an
\RIG defines a structure similar (but different) to rectilinear
convex-hull, which seem to have not been studied before.

\paragraph*{Our results}
\begin{compactenumI}

    \smallskip%
    \item \textsc{Box hull}. We fully characterize the union of
    rectangles associated with $\RIG$, and give simple algorithms to
    compute it, find a set of $O(n)$ interior disjoint rectangles that
    partition it, and, given a point in the union, find a rectangle
    that contains it.  We refer to this region induced by the set of
    points $\P$ as the \emphi{box hull} of $\P$, see
    \figref{two:diags} and \figref{b:h:proof}.  See
    \secref{box:hull} for details.

    The box hull is a superset of the rectilinear convex-hull of a
    point set, and has the advantage of being connected and $x$ and
    $y$ monotone. To the best of our knowledge this concept was not
    studied before, and it seems to be natural and more convenient
    than previous similar concepts.

    \smallskip%
    \item \textsc{Implicit computation of the \RIG}. We show that the
    graph $\RIGraph = \RIGX{\P}$ has a representation as the (edge)
    disjoint union of $O( n )$ bicliques, where the total weight of
    the bicliques is $O(n \log^2 n)$ (and it can be computed at this
    time). As a reminder, the biclique $\L \otimes \R$ defined by
    disjoint sets $\L, \R$, is the complete bipartite graph having
    $\L \cup \R$ as its set of vertices, and edges are only between
    the two parts. The weight of such a biclique is
    $\cardin{\L} + \cardin{\R}$. This implies that while $\RIGraph$
    might have quadratic number of edges, it still has a compact
    representation of near linear size.

    To put this result in prospective, it is relatively not hard to
    observe that such a decomposition can exist by using orthogonal
    range-searching data-structures. This however seems to lead to
    bounds of overall size $O(n \log^{3} n)$. Getting rid of the extra
    logarithmic factor in total weight requires taking a more
    elaborate approach which is significantly more
    involved. Similarly, reducing the number of bicliques to linear
    requires some additional ideas. We consider this to be the main
    contribution of our work, and it is described in
    \secref{biclique:cover}.

    \smallskip%
    \item \textsc{Lower bound on the weight of the implicit
       representation.} We prove a lower bound of $\Omega(n\log n)$ on
    the weight of the biclique cover of an \RIG{}. Thus, our
    construction above seems to be optimal (up to potentially a single
    logarithmic factor). See \secref{lower:bound} for details.

    \smallskip%
    \item \textsc{Depth approximation.} Using the small biclique cover
    we show, for the set of rectangles associated with an \RIG{}, a
    construction of a data structure for $(1+\eps)$-approximate-depth
    queries, whose preprocessing phase can be used as a deepest-point
    $(1+\eps)$-approximation algorithm. While the underlying set of
    rectangles might have quadratic size, the preprocessing and space
    used by our data-structure are near linear. Furthermore, queries
    can be answered in logarithmic time. Interestingly, to this end
    we replace every biclique by a set of weighted rectangles whose weight approximates the weight induced by the rectangles of
    the biclique. By approximating ``levels'' in the bicliques, and
    using exponential scales and the structure of the bicliques, we
    are able to do this in near linear time and space, critically
    using the implicit low-weight biclique representation of the \RIG{}.

\end{compactenumI}

\bigskip%
\noindent%
A few additional minor new results:

\begin{compactenumI}[resume]
    \smallskip%
    \item In \apndref{random:points} we give a tight bound of
    $\Theta( n \log n)$ on the number of edges in the \RIG, with high
    probability, for a set of $n$ points picked uniformly at random. A
    concentration result was already shown by Chen \etal
    \cite{cpst-dgpsp-09}, but our bound holds with high probability.

    \smallskip%
    \item In \apndref{structure} we prove that the complexity of
    \RIG{}s is solely due to bicliques. We do so by proving various
    structural results on \RIG. For example, a presence of a biclique with both
    sides having size at least $k$ implies the existence of three
    ``parallel'' long chains that together form this
    biclique. Similarly, the \RIG can not have large (regular)
    cliques as subgraphs, and any three-sided clique must have a side that
    contains only a single vertex.

\end{compactenumI}

\myparagraph{Sketch of computation of the biclique cover of the \RIG}
Using multi-level orthogonal range searching one can compute for each
point $\p \in \P$ the set of points that it dominates, where the set
is represented implicitly as nodes of the data-structure, and every
node is associated with all the points stored in its subtree.  For our
purposes, we need the maxima of each such set, and we need to stitch
these together to form a single maxima chain (which is the set of all
points that support an empty rectangle with $\p$ appearing as the top
right corner). Note that a similar procedure is required for the
anti-domination relation. To this end, one can compute for each
canonical set in the data-structure its maxima, and perform the
stitching quickly, see \secref{stitching}. This yields a
representation of all such points that are neighbors of $\p$ in the
\RIG{} over $\P$ using $O( \log^3 n)$ sets and readily leads to a
biclique representation of the \RIG of weight $O( n \log^3 n)$.

To improve this we need to ``open'' the orthogonal range-searching
data-structure. We thus consider the variant where the point $\p$ has
a higher $y$-coordinate than all the points of $\P$ (we refer to $\P$
as being a $y$-set in such a case). This turns out to be a more
``one-dimensional'' problem. We present a stack data-structure that
enables us to store and preprocess such queries, where instead of
$O( \log^2 n)$ canonical sets in the output we are able to reduce
this number to $O( \log n)$, thus getting the desired improvement. The
one dimensional data-structure requires a non-trivial modification of
the Bently-Saxe technique coupled with persistence, see \secref{stack}.

The number of bicliques this generates is still $O(n \log^2 n)$. To
reduce the number of bicliques further, being somewhat imprecise, we
modify the above stack data-structure so that it only generates either
``singleton'' bicliques or bicliques that have $\Omega(\log n)$
elements on each side (i.e., ``heavy bicliques''). The singleton
bicliques get merged into star bicliques which results in only $O(n)$
bicliques.

\section{Preliminaries}

\begin{defn}
    For two points $\p, \q \in \Re^2$, let $\rA = \rectY{\p}{\q}$
    denote the minimal axis-parallel closed rectangle that contains
    $\p$ and $\q$. Note that $\p$ and $\q$ form antipodal corners of
    $\rA$. The points $\p$ and $\q$ \emphi{support} $\rectY{\p}{\q}$,
    and $\rectY{\p}{\q}$ is the \emphi{rectangle of influence} of $\p$
    and $\q$.
\end{defn}

\begin{defn}
    For a set $\P$ of $n$ points in $\Re^2$. The \emphi{rectangle
       influence graph} (\emphcolor{\RIG}) of $\P$ is the graph
    $\RIGX{\P} = (\P,\EE)$ where $\p \q \in \EE$ $\iff \rectY{\p}{\q}$
    does not contain any point of $\P$ in its interior. $\RIGX{\P}$ is
    also known as the Delaunay graph of $\P$ with respect to axis parallel
    rectangles \cite{cpst-dgpsp-09}, or negative witness rectangle
    graph \cite{adh-wrg-14}. Similarly, the $\RIG_{\!\leq k}$ graph of
    $\P$ is $\RIGKY{k}{\P}$ = $(\P, \EE_{\leq k})$, where
    $\p \q \in \EE_{\leq k}$ $\iff \rectY{\p}{\q}$ contains at most
    $k$ points of $\P$ in its interior.
\end{defn}

\begin{defn}
    A point $\p$ \emphi{dominates} (resp. \emphw{anti-dominates}) a
    point $\q$, denoted by $\q \domby \p$ (resp. $\p \adom \q$), if
    $x(\q) < x(\p)$ and $y(\q) <y(\p)$ (resp. $x(\p) < x(\q)$ and
    $y(\p) > y(\q)$).
\end{defn}

A $\domby$-\emphi{chain} is a sequence of points
$\p_1, \ldots, \p_k \in \P$, such that
$\p_1 \domby \p_2 \domby \cdots \domby \p_k$. Here, an
$\domby$-\emphi{antichain} corresponds to an $\adom$-chain.

\begin{defn}
    The $\radom$-chain made of $\domby$-maximal elements in $\P$ is
    the \emphi{$\domby$-maxima} of $\P$, denoted by $\maximaX{\P}$,
    which is a sequence of points ordered from left to right. The
    concept $\amaximaX{\P}$ is defined in a similar fashion. See
    \figref{maxima}.
\end{defn}

\NotSoCGVer{%
\begin{figure}[h]
    \begin{tabular}{*{3}{c}}
      \includegraphics{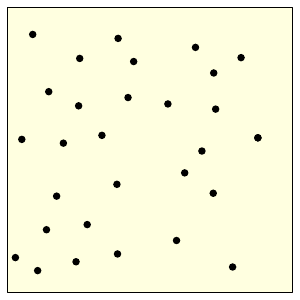}
      &
        \includegraphics[page=2]{figs/maxima}
      &
        \includegraphics[page=3]{figs/maxima}
      \\
      $\P$
      &
        $\maximaX{\P}$
      &$\amaximaX{\P}$
    \end{tabular}
    \caption{The $\domby$-maxima and $\adom$-maxima of a point set.}
    \figlab{maxima}
\end{figure}
}%
\SoCGVer{%
\begin{figure}[h]
    \begin{tabular}{*{3}{c}}
      \includegraphics[width=0.3\linewidth]{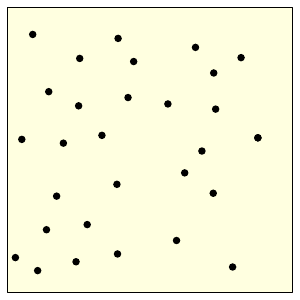}
      &
        \includegraphics[page=2,width=0.3\linewidth]{figs/maxima}
      &
        \includegraphics[page=3,width=0.3\linewidth]{figs/maxima}
      \\
      $\P$
      &
        $\maximaX{\P}$
      &$\amaximaX{\P}$
    \end{tabular}
    \caption{The $\domby$-maxima and $\adom$-maxima of a point set.}
    \figlab{maxima}
\end{figure}
}%

The graph whose edge set is formed only by rectangles supported by a
pair of points where one dominates (resp. anti-dominates) the other
are denoted by $\RIGDX{\P}$ (resp. $\RIGAX{\P}$). These two graphs are
edge disjoint and $\RIGX{\P} = \RIGDX{\P} \cup \RIGAX{\P}$.

\section{Biclique cover of \TPDF{$\RIGX{\P}$}{RIG(P)}}
\seclab{biclique:cover}

\subsection{Definitions}

A \emphi{biclique} over (disjoint) sets $X$ and $Y$ is the complete
bipartite graph
\begin{math}
    X \otimes Y =%
    \bigl(X \cup Y, \Set{ x y }{x \in X, y \in Y}\bigr).
\end{math}

\begin{defn}
    For a graph $\G = (\VV, \EE)$, a sequence $\BC$ of pairs of sets
    $\{\R_1,\B_1\}, \ldots, \{\R_t, \B_t\}$ is a \emphi{biclique
       cover} of $\G$ if, for all $i < j$, we have
    \begin{compactenumi}
        \item $\R_i, \B_i \subseteq \VV$,
        \item $\R_i \cap \B_i = \emptyset$,
        \item
        $\EGX{\R_i \otimes \B_i} \cap \EGX{\R_j \otimes \B_j} =
        \emptyset$, and
        \item $\cup_k (\R_k \otimes \B_k) = \G$.
    \end{compactenumi}
    \smallskip%
    Thus, for every edge $rb \in \EGX{\G}$ there exists a unique
    index $i \in \IRX{t}$ such that $r\in \R_i, b\in \B_i$.

    The \emphi{weight} of $\BC$ is
    $\WeightX{\BC} = \sum_{i=1}^t \pth{ \cardin{\R_i } +
       \cardin{\B_i}}$.
\end{defn}
The weight of a biclique cover is the space required to store the
graph (implicitly) as a list of bicliques, where each biclique is
specified by listing its vertices.

\newcommand{\NbrX}[1]{\Gamma\pth{#1}}%

\begin{example}
    Note, that any graph $\G=(\VV,\EE)$ with $n$ vertices
    $\VV =\{ v_1, \ldots, v_n\}$, has a biclique cover with $n-1$
    cliques, with the $i$\th biclique having $\{v_i\}$ as one side,
    and all its neighbors in $V^{i+1} =\{v_{i+1},\ldots, v_n\}$ on the
    other side.  Formally,
    $\G = \cup_i \bigl( \{v_i\} \otimes (V^{i+1}\cap
    \NbrX{v_i})\bigr)$, where $\NbrX{v}$ denotes the set of neighbors
    of $v$ in $\G$.  Of course, the weight of this cover is
    (asymptotically) the number of edges in the graph, thus the
    challenge is to compute a cover of a graph with a near-linear
    number of bicliques and small total weight. A random graph where
    each edge is chosen with probability half readily shows that there
    are graphs where the weight is at least quadratic in any biclique
    cover. Namely, graphs in general do not have near-linear weight
    biclique cover.
\end{example}

\subsection{Stitching maximas quickly}
\seclab{stitching}

Consider a (balanced) binary search tree $\Tree$ that stores a
sequence $S =\pth{s_1\prec s_2 \prec \cdots \prec s_n}$ in its leaves
according to some linear order $\prec$. For a node $v \in \Tree$, let
$L(v)$ denote the sorted list of elements stored in its subtree.

\begin{observation}
    \obslab{interval}%
    For any two elements $s_i, s_j \in S$, with $i <j$, there are
    $t \leq 2\cdot\mathrm{height}(\Tree) -1$ nodes $v_1, \ldots v_t$ of
    $\Tree$, such that
    $s_i, \ldots s_j = L(v_1) \cup \cdots \cup L(v_t)$.  The sequence
    $v_1, \ldots, v_t$ is the implicit \emphi{representation} of the
    (explicit) sequence $s_i, \ldots, s_j$.  The nodes
    $v_1, \ldots, v_t$ appear on or adjacent to the two paths in
    the tree from the root to $s_i$ and $s_j$, so
    computing this representation takes $O( \heightX{\Tree} )$ time.
\end{observation}

\begin{lemma}
    \lemlab{chain-union}%
    Let $\rect_1$ and $\rect_2$ be two disjoint rectangles with their
    respective maximas $\L_1 =\maximaX{\rect_1 \cap \P}$ and
    $\L_2 = \maximaX{\rect_2 \cap \P}$ stored in balanced binary
    search trees $T_1$ and $T_2$ respectively (say, the points are
    stored sorted in increasing $x$-axis order).  Then, one can
    compute, in $O( \log n)$ time, $m = O( \log n)$ nodes
    $v_1 ,\ldots, v_m$, such that
    $L(v_1) \concat \cdots \concat L(v_m)$ is
    $\L = \maximaX{(\rect_1\cup\rect_2) \cap \P}$, where $\concat$
    denotes concatenation and $n = \cardin{\P}$.  Specifically, if
    $\L_1$ and $\L_2$ are represented by lists of nodes of size
    $\leq t$, then $\L$ can be computed in $O( t + \log n)$
    time\footnote{Assuming that the trees used in the representations
       have depth $O(\log n)$.}.
\end{lemma}

\begin{proof}
    The algorithm checks if $\rect_1$ and $\rect_2$ are separated by a
    horizontal line, a vertical line, or both. The various cases are
    handled in similar fashion, so we describe only the case that
    $\rect_1$ and $\rect_2$ are horizontally separated, and $\rect_1$
    is vertically above $\rect_2$, see \figref{merge:maxima}.

    \begin{figure}[ht]
        \phantom{}%
        \hfill%
        \includegraphics[page=1]{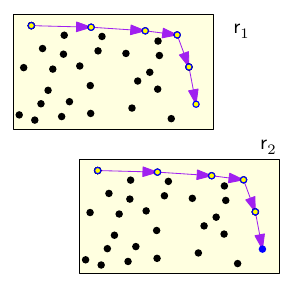} \hfill%
        \includegraphics[page=2]{figs/maxima_merge} \hfill\phantom{}%
        \caption{}
        \figlab{merge:maxima}
    \end{figure}

    Let $\L_i = \maximaX{\rect_i \cap \P}$ for $i \in \{ 1,2\}$. Let
    $x$ be the $x$-coordinate of the last point of $\L_1$, and
    observe that the merged maxima is the result of throwing away all
    the points of $\L_2$ that have smaller $x$-coordinate than $x$
    from $\L_2$ and concatenating the two sequences. Namely, the
    merged maxima is $\L_1 \cup (\L_2 \cap (x,\infty))$, where
    $\L_2 \cap (x,\infty)$ denotes the desired suffix of $\L_2$.  By
    \obsref{interval} this suffix has the desired representation as a
    sequence of $O( \log n)$ nodes. Prefixing this sequence by the
    root of the tree representing $L_1$ yields the desired result.

    The second claim follows by observing that computing a subsequence
    stored in a node $v$, between two given elements, can be done in
    $O( \log n)$ time, if implicit representation is sufficient. It is
    easy to verify that the split and concatenation operations
    required by the above, can be implemented by computing a
    subsequence, see \obsref{interval}, which might yield an implicit
    representation of size $O( t + \log n)$, and then concatenating
    this nodes list with the list of nodes representing the other
    set. In either case, this can be done in $O( t + \log n)$ time.
\end{proof}

\subsection{A helper data-structure: Stack with range queries}
\seclab{stack}

In the following, we describe a stack data-structure that stores a
sequence of elements (e.g. points) with an associated real value
(e.g. the points' $x$ coordinates).  Note that the elements must be
inserted in increasing order.  The following operations are supported:
\begin{compactenumi}
    \item $\pushOp(x)$: Push $x$ to the stack.

    \smallskip%
    \item $\popOp(k$): Pops the top $k$ elements from the stack.

    \smallskip%
    \item $\reportOp(x,y)$: Output a representation of the subsequence
    of the elements stored in the stack between $x$ and $y$.
\end{compactenumi}
\smallskip%
In addition, the data-structure can declare that a specific list/set
of elements is \emphi{canonical}. For a list of size $t$, making it
canonical costs $O(t)$ time. Importantly, the above \reportOp output
is a list of canonical sets.

\myparagraph{Implementation}
We follow the Bently and Saxe \cite{bs-dspsdt-80} technique.  The
basic building block would be a perfect binary tree. Such a tree $T$
storing $2^t$ elements has \emphi{rank} $t$ (i.e., height), where all
the values are stored in the leafs of $T$. Let $\listX{T}$ denote the
list/sequence of elements stored in $T$.

At any given point in time the stack contains a list of $O(\log n)$
such trees $T_1,...,T_m$, referred to as the current \emphi{forest},
with $\listX{T_1} \concat \ldots \concat \listX{T_m}$ being the
content of the stack (i.e., the top of the stack is the last element).
The following invariants are maintained:
\begin{compactenumI}
    \smallskip%
    \item $|T_i| \geq |T_{i+1}|$, for all $i$.

    \smallskip%
    \item There are at most two trees with the same rank in this
    list\footnote{Unlike the typical implementation, where there is at
       most one tree of each rank in this list. This ``lazyness'' is
       critical for our purposes.}.
\end{compactenumI}
\smallskip%
The specific operations are implemented as follows:
\smallskip%
\begin{compactenumI}[leftmargin=0.25cm]
    \item $\pushOp(x)$: A tree of rank zero containing this element
    and appending it to the end of the list of trees in the
    stack. Now, if the last three trees in the stack have the same
    rank, say $t$, the first two trees are merged into a new tree of
    rank $t+1$, and the new tree replaces the old two trees. This is
    repeated till no three trees in the list of the stack are
 	of the same rank (i.e., the above invariants are
    maintained). Whenever a new tree is being created it is being
    canonized (i.e. registered as a canonical set). This is
    implemented in a mundane way during merge -- the new tree is
    formed by (deep) copying its two subtrees, and by creating a new
    root node.

    \smallskip%
    \item $\popOp(k)$: The data-structure locates the tree $T_i$ in
    the forest containing the $k$\th value from the end of the
    stack. It then removes the following trees $T_j$, with $j > i$
    from the forest. Let $r = \rankX{T}$.  The data-structure breaks
    $T_i$ into $u<r$ trees $T_1',\ldots, T_u'$, such that
    $\listX{T_1'} \concat \cdots \concat \listX{T_1u'}$ are all the
    elements appearing before $x$ in $T_i$. The data-structure
    replaces $T_i$ in the forest by $T_1',\ldots, T_u'$. Note, that as
    $\rankX{T_1'} > \cdots > \rankX{T_u'}$, this preserves the
    invariants specified above. Furthermore, $T_1', \ldots, T_u'$ are
    trees that were created earlier by the data-structure, so
    $\popOp(k)$ is implemented in $O( \log n)$ time.

    \smallskip%
    \item $\reportOp(x,y)$: Scanning the forest, the algorithm
    computes the range of trees $T_\alpha, \ldots, T_\beta$ that their
    lists contain the desired elements. Using binary search on
    $T_\alpha$ and $T_\beta$, the algorithm computes $O( \log n)$
    nodes (an internal node of tree can be treated as its own tree)
    that lists the elements that are in the desired range. Let
    $\Forest_\alpha$ and $\Forest_\beta$ be these forests computed for
    $T_\alpha$ and $T_\beta$ respectively. Clearly, the desired output
    is
    $\Forest_\alpha \concat T_{\alpha+1} \concat \cdots \concat
    T_{\beta-1} \concat \Forest_\beta$, a list of
    $O( \log n)$ trees, and it can be computed in $O( \log n)$ time.

\end{compactenumI}

\subsubsection{Performance}
\begin{lemma}
    \lemlab{stack}%
    The above data structure performs a sequence of $n$ \pushOp/\popOp
    operations in $O(n\log n)$ time, and this also bounds the total
    size of the canonical sets.  The total number of canonical sets is
    $O( n )$.  A reporting query takes $O(\log n)$ time, and returns
    $O( \log n)$ canonical sets per query.
\end{lemma}

\begin{proof}
    In order to bound the total size of the trees created by $n$
    operations, consider the set of trees of rank $k$ created during
    the lifetime of the data structure. Let $o_1,...,o_n$ be the
    sequence of operations performed.  The operation $o_t$ is the
    \emph{step} at time $t$. A tree $T$ of rank $k$ was created at
    time $t$ if $o_t$ is the push operation that resulted in the
    creation of $T$ by merging two trees of rank $k-1$.

    We claim that the time difference between the creation of two
    trees $T,T'$, both of rank $k$, at times, say, $t$ and $t'$
    respectively where $t < t'$. is $\geq 2^{k-1}-1$. At time $t-1$,
    i.e. just before $o_t$ was executed, the current forest had two
    trees of rank $k-1$, after the push operation $o_t$, a cascade of
    merges of trees resulted in the creation of $T$. Note that a merge
    of two trees of rank $u$ is triggered if the forest has three
    trees of rank $u$. After the merge, there is exactly one tree in
    the forest of rank $u$. Namely, immediately after $o_t$ was
    executed, the forest has exactly one tree of rank $u$ for
    $u =0,\ldots, k-1$.

    If the sets of values stored in $T$ and $T'$ are disjoint, it is
    clear that at least $2^k/2$ more insertions are required in order
    for another tree of size $2^k$ to be created even if no pop
    operations were executed between times $t$ and $t'$. If the sets
    of values stored in $T$ and $T'$ are not disjoint, then consider
    the first time $t''$, $t < t'' < t'$, in which a value stored in
    $T$ is popped alongside all elements inserted after it.  This
    ``shatters'' the tree $T$ in the forest, replacing it by a
    sequence of subtrees. Importantly, all these subtrees have each
    unique rank, and there are no other trees in the forest with this
    rank (at this time). Specifically, after $o_{t''}$ the forest has
    at most one tree of rank $k-1$, and the same holds for smaller
    ranks. In particular, the total size of the elements stored in
    these ``tail'' trees is at most $\gamma = 2^{k}-1$. Thus, at least
    $3\cdot 2^{k-1}-\gamma = 2^{k-1}-1$ more elements have to be
    pushed into the stack before the forest would contain three blocks
    of rank $k-1$ and trigger the merge creating $T$, as claimed.

    Thus, the number of trees of rank $k$ created during the course
    of $n$ operations, is $O(n/2^k)$, and the overall size of all the
    canonical sets is at most
    $\sum_{k=1}^{\log n} 2^k O( n/2^k) = O(n\log n)$.
\end{proof}

\begin{remark}[Persistence]
    A minor technicality is that we need to perform the
    $\reportOp(x,y)$ query on the stack at the point in time just
    before a value larger than $y$ was inserted. To this end, after
    each operation the data-structure creates a copy of the forest --
    as the forest is simply a list of $O( \log n)$ trees, this can be
    done in $O( \log n)$ time. Now, given a query $\reportOp(x,y)$,
    the algorithm first does a binary search to find the copy of the
    forest at the ``right'' time, and then performs the query on this
    version. A more efficient persistence scheme might be possible
    here, but it is  irrelevant for our purposes.
\end{remark}

\subsection{Extracting the maxima}

\begin{lemma}
    \lemlab{smaller:bi:c:cover}%
    Let $\P$ be a set of $n$ points in the plane in general position.
    The graph $\RIGX{\P}$ has a biclique cover with $O( n\log n )$
    bicliques, and total weight $O(n\log^2 n)$. The biclique cover can
    be computed in $O(n\log^2n)$ time.
\end{lemma}

\begin{proof}
    We describe the construction for $\RIGDX{\P}$. A similar
    construction applies to $\RIGAX{\P}$, and together they form the
    desired cover. Let $\Y = \Y(\P)$ be the set of $y$-coordinates of
    the points of $\P$. Let $\Tree_\Y$ be a balanced binary search
    tree on the points of $\Y$. A \emphi{$y$-set} is the set of
    (original) points of $\P$ stored in a subtree of $\Tree_\Y$.
    Given an unbounded downward ray $\rho$ on the $y$-axis, the set
    $\rho \cap \Y$ can be represented as the disjoint union of
    $\ceil{ \log_2 n}$ $y$-sets.

    Consider such a $y$-set, and its corresponding subset
    $Z \subseteq \P$. Assume the points of $Z$ are (pre)sorted by the
    $x$-axis, and $Z = \{ \p_1, \ldots, \p_u\}$. The algorithm inserts
    the points of $Z$ in this order into the stack data structure
    described in \lemref{stack} while maintaining
    $\maximaX{\p_1,\ldots, \p_i}$. Namely, before $\p_i$ is
    pushed, the algorithm pops from the top of the stack all the
    points that $\p_i$ dominates. The algorithm performs this
    construction on all the $y$-sets of $\Tree_\Y$. The total size of
    the $y$-sets is $O( n \log n)$. Presorting the points of $\P$ by
    the $x$ and $y$-axis, and computing these data-structures in a
    top-down recursive fashion on $\Tree_\Y$ can be done in
    $O(n \log^2 n)$ time overall. The total number of the canonical
    sets created in all the stacks is $O( n \log n)$.

    Now, for each point $\p \in \P$ the algorithm computes the
    maxima of the points it dominates. Indeed, for a point $\p$, the
    algorithm computes the $O(\log n)$ disjoint $y$-sets whose
    union is all the points of $\P$ with $y$ coordinate smaller than
    $\p$. For each such $y$-set the algorithm precomputed the
    data-structure of \lemref{stack}, and it can extract the maxima in
    the $y$-set up to the point $\p$.  Using the stitching algorithm
    of \lemref{chain-union} one can find the maxima of the points that
    $\p$ dominates. This maxima would be represented by the disjoint
    union of $O( \log^2 n)$ canonical sets.

    In detail, given a point $\p = (x,y) \in \P$ we compute a disjoint
    union of $m=O( \log n)$ $y$-sets $Y_1, \ldots, Y_m$ such that
    $\cup_i Y_i$ is the set of all points with $y$-coordinate smaller
    than $\p$. Assume that all the points of $Y_i$ have bigger $y$
    coordinates than the points of $Y_{i+1}$, for all $i$.  Let
    $\Quad = (-\infty, x) \times (-\infty, y)$. Clearly, all the
    Delaunay rectangles having $\p$ on the right-top corner have a
    point of $\maximaX{\Quad \cap \P}$ on the other corner, and also
    $\Quad \cap \P = \cup_i (\Quad \cap Y_i)$. Let
    $\listSX{1}' = \maximaX{ \Quad \cap Y_1}$ -- this can be computed
    in $O( \log n)$ time using the stack data-structure computed for
    $Y_1$, and let $r_1$ be the $x$-coordinate of the last point of
    $\listSX{1}$.  Naively, one can repeat this process for all $i$,
    computing $\listSX{i}'= \maximaX{ \Quad \cap Y_i}$. The desired
    maxima is the ``stitched'' maxima of these maximas. It is more
    efficient to do this stitching process directly. Indeed, let
    $\listSX{i} = \maximaX{\cup_{k=1}^i \listSX{k}'}$, with
    $\listSX{1} = \listSX{1}'$. Let $r_i$ be the $x$-coordinate of the
    rightmost point in $\listSX{i}$. The algorithm now adds the next
    portion of the maxima that uses points of $Y_{i+1}$. To this end,
    using persistence, the data-structure recovers the last time, say
    $t_{i+1}$, when a point of $Y_{i+1} \cap \Quad$ was inserted into
    its stack. Let $S_{i+1}$ denote this snapshot of this stack at
    time $t_{i+1}$. The algorithm then extracts the maxima of
    $Y_{i+1}$ appearing in $\listSX{i+1}$ by performing the query
    $\reportOp(r_i, x)$. The algorithm repeats this process for the
    next value of $i$ till the whole maxima is extracted. This takes
    $O( \log^2 n)$ time and results in a representation of
    $\maximaX{\Quad \cap \P}$ using $O( \log^2 n )$ canonical sets.
    We register $\p$ with all of these canonical sets.

    At the end of this process, every canonical set $X$ has a set of
    points $Y$ registered with it. The algorithm outputs $X \otimes Y$
    as one of the computed bicliques. The bound on the size of the
    bicliques is immediate -- the total size of the canonical sets is
    $O(n \log^2n)$, and each point of $\P$ is registered with
    $O( \log^2 n)$ canonical sets. As for their total number,
    \lemref{stack} implies that the number of canonical sets every
    $y$-set induces is proportional to its cardinality. Thus, the
    overall number of canonical sets is $O(n \log n)$. Every canonical
    set $X$ has a set of $Y$ of (query) points registered with it, and
    give rise to the biclique $X \otimes Y$. Thus, the total number of
    bicliques is $O(n \log n)$.
\end{proof}

Being a bit more careful, one can reduce the number of bicliques to
linear while keeping the total weight (asymptotically) the same.
\begin{theorem}
    \thmlab{smaller-biclique:2}%
    Let $\P$ be a set of $n$ points in the plane in general position.
    The graph $\RIGX{\P}$ has a biclique cover with $O( n )$
    bicliques, and total weight $O(n\log^2 n)$. The biclique cover can
    be computed in $O(n\log^2n)$ time.
\end{theorem}
\begin{proof}
    The idea is two folds.  The basic idea is that bicliques involving
    a singleton on one side are fine -- there are going to be only $n$
    such bicliques globally.  Similarly, one can argue that the number
    of large bicliques involving a large canonical set (roughly of
    size $\Omega(\log^2n)$) on one side is globally bounded by $O(n)$.

    The problem is thus with middling size canonical sets of
    size roughly $O( \log n)$. We first describe how to reduce the number of
    canonical sets used by the stack data-structure. The idea is to
    add a buffer of size (say) $\tau = 3 \ceil{\log n}$ of
    singletons that will be stored as they are. Whenever the buffer
    is filled the algorithm extracts the last $\ceil{\log n}$
    elements, and create a canonical set for these elements, deleting
    these sets from the buffer (which is a FIFO queue). It is easy to
    verify that the previous stack construction can be easily modified
    to work with this buffer thus creating canonical sets only of size
    $\geq \tau$. The number of non-singleton canonical sets
    created by the stack is therefore
    \begin{equation*}
        \sum_{k= \floor{\log \tau}} n/2^k = O(n /\tau) = O(n /
        \log n).
    \end{equation*}

    We apply the same idea to the $y$-strips. We only build the stack
    data-structure for strips that contain more than $\tau$
    points. Formally, we partition the point sets into strips of size
    $\tau$ in the $y$-order, and build balanced binary tree on these
    strips. Now, a query would involve $O( \log n)$ $y$-sets (each of
    size at least $\tau$) where potentially the smallest strip might
    contain $\tau$ points, and the query quadrant intersects it only
    partially. The query in this top strip is answered directly by
    scanning (i.e., we compute the maxima of the points in this strip
    inside the query quadrant and register the query point with each
    point on the maxima), while for any other of the $y$-strips we use
    the stitching algorithm described in
    \lemref{smaller:bi:c:cover}. Clearly, as before, there are
    $O( \log n )$ $y$-strips involved with the query and each one
    returns $O( \log n)$ canonical sets, thus reproducing the old
    bound on the total weight of the bicliques.

    As for the overall number of bicliques, observe that every
    biclique rises out of a canonical set in one of the (modified)
    stack data-structures. If a $y$-set has $t$ points then the total
    number of canonical sets it defines is $O( t/ \log n)$ (ignoring
    singletons, since these ones are already counted directly). Thus, since
    the total number of $y$-strips of size $2^i \tau$ is
    $n/(2^i \tau)$ and each one of them contributes
    $O(2^i \tau/ \log n) = O(2^i)$ (non-singleton) canonical sets, we
    conclude that the total number of (non-singleton) canonical sets
    over all $y$ strips is
     \begin{equation*}
         O \Bigl( \sum\nolimits_{i=0}^{O( \log n)} \frac{n}{2^i \tau} \cdot  2^i
         \Bigr)
         = O(n).
         \SoCGVer{\qedhere}
     \end{equation*}
\end{proof}

\begin{observation}
    \obslab{rects:implicit}%
    Let $\P$ be a set of points in the plane, and let
    $\Rects = \RectsInfX{\P}$, and let $\BC$ be the biclique cover
    computed by \thmref{smaller-biclique:2}. Each biclique corresponds to
    two disjoint point sets $\P_i, \PSA_i \subseteq \P$, such that
    $\Rects = \cup_{i} \rectY{\P_i}{\PSA_i}$, where
    $\rectY{\P_i}{\PSA_i} = \cup_{\p \in \P_i, \q \in \PSA_i} \{
    \rectY{\p}{\q} \}$. Furthermore $\P_i$ and $\PSA_i$ are separable
    by both a horizontal and vertical lines. Finally, we can assume
    that each of the sets $\P_i, \PSA_i$ are available as a sequence sorted by
    both $x$-axis and $y$-axis order. This provides us with an
    implicit efficient representation of $\Rects$ that can be computed
    in $O( n \log^2n )$ time, made out of $O(n)$ pairs, and of total
    weight $O(n \log^2 n)$.
\end{observation}

\begin{corollary}\RefProofInAppendix{k:depth}%
	\corlab{k:depth}%
	The graph $\RIGKY{k}{\P}$ has a biclique cover of size $O(k
        n)$. Furthermore, the total weight of the cover is
        $O(k n\log^2 n)$, and it can be computed in $O(k n\log^2n)$
        time.
\end{corollary}

\begin{proof:e}{{\Xcorref{k:depth}}}{k:depth}
    By creating $k$ stack data structures for every canonical slab we
    can easily maintain the first $k$ levels of, say $\domby$-maxima,
    each contained in a separate stack. An introduction of a new point
    $\p$ becomes a $k$-step process starting with the insertion of the
    point to the upper-most stack maintaining the maxima chain, but
    now instead or simply deleting the chain prefix that is now in the
    second level of maxima, that sequence of points is now inserted
    simultaneously to the next stack. This is possible since $\p$,
    being inserted as the rightmost point seen so far, dominates a
    continuous prefix of every level of the $k$ $\domby$-maximas, and
    thus we can, for every level, report and remove the points
    dominated by $\p$ using the operations as they are described in
    \secref{stack}, and insert the points from the previous level by
    using an individual insert operation for every point. Notice that
    every point $\q_{i+1}$ that remains in the $(i+1)$\th level after
    the deletion step is to the left of any point $q_i$ inserted
    (i.e. deleted from the $i$\th level) since otherwise we get that
    $q_i\domby \q_{i+1}$ which is a contradiction. Since every point
    is inserted and deleted exactly once at any level the size and
    time complexity are $O(k n\log^2 n)$ as stated.
\end{proof:e}

\subsection{A lower bound}
\seclab{lower:bound}

\begin{theorem}{\cite{bs-ss-07}}
    \thmlab{bi:c:k:n}%
    The weight of any biclique cover of $K_n$ is
    $\Omega(n \log n)$\NotSoCGVer{%
       \footnote{Proof sketch: Delete the vertices on one side of each
          biclique, randomly choosing the side. This leaves at most
          one vertex that is not deleted. On the other hand, if a
          vertex participates in $k$ bicliques, its probability of
          survival is $1/2^k$. If follows that the average vertex
          participates in $\Omega(\log n)$ bicliques.}}, where $K_n$
    denotes the complete graph over $n$ vertics.
\end{theorem}

\begin{figure}
    \NotSoCGVer{%
       \includegraphics[page=1]{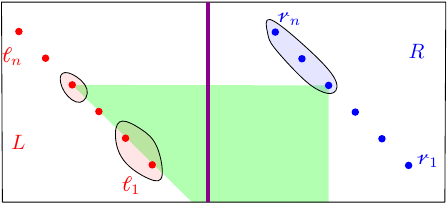}%
       \hfill%
       \includegraphics[page=2]{figs/bcc_lower_bound} }%
    \SoCGVer{%
       \includegraphics[page=1,width=0.46\linewidth]{figs/bcc_lower_bound}%
       \hfill%
       \includegraphics[page=2,width=0.46\linewidth]{figs/bcc_lower_bound}
    }

    \caption{Illustration of the construction in
       \lemref{bcc-lower-bound}. Left: the set of $2n$ points $\R$ and
       $\B$ and a biclique $\{A_i, B_i\}$. Right: An illustration of
       the construction of $\{A'_i,B'_i\}$ as a projection to a
       1-dimensional space}
    \figlab{lower:bound}%
\end{figure}

\begin{lemma}
    \lemlab{bcc-lower-bound}%
    There exists a set $\P$ of $n$ points, such that any biclique
    cover of $\RIGX{\P}$ has weight $\Omega(n\log n)$.
\end{lemma}

\begin{proof}%
    Let $\L$ and $\R$ be two sets of points, defined by
    \begin{equation*}
        \L= \Set{\pl_i = (-i, i)}{i\in \IRX{n}}
        \qquad\text{and}\qquad%
        \R = \Set{\pr_i = (2n - i, i)}{i\in \IRX{n}}.
    \end{equation*}
    See \figref{lower:bound}.  Let $\P = \L \cup \R$,
    $\RIGraph=\RIGX{\P}$, and $\BC$ be any biclique cover of
    $\RIGraph$. Next, consider the restriction of the biclique cover
    to the two sides (some edges would no longer be covered):
    \begin{equation*}
        \BCA
        =
        \Set{
           (X \cap \L) \otimes (Y \cap \R),
           (X \cap \R) \otimes (Y \cap \L)
        }{X \otimes Y \in \BC}.
    \end{equation*}
    Observe that $\WeightX{\BC} = \Theta( \WeightX{\BCA})$.  We
    rewrite $\BCA = \{ \L_1 \otimes \R_1,\ldots, \L_s\otimes\R_s \}$,
    where $\L_i \subset \L$, and $\R_i \subset \R$, for all $i$.  For
    all pairs $\L_i \otimes \R_i \in \BCA$ and all points
    $\pl_\lambda = (-\lambda, \lambda) \in \L_i$ and
    $\pr_\rho = (2n-\rho, \rho) \in \R_i$, we have that
    $\lambda \leq \rho$, as otherwise $\lambda > \rho$, and the
    associated rectangle $\rect =\rectY{\pl_\lambda}{\pr_{\rho}}$
    contains the point
    $\pl_{\lambda-1} = (-\lambda+1,\lambda-1) \in \L$, which implies
    that $\rect \notin \RectsInfX{\P}$. Namely,
    $\pl_\lambda \pr_\rho \notin \EGX{\RIGraph}$ for $\lambda > \rho$.

    For a set of points $Z$, let $Y(Z) = \Set{y }{ (x,y) \in Z}$ be
    the set of $y$-coordinates appearing in points of $Z$. Let
    $B_i' = Y(\L_i)$, $T_i' =Y(\R_i)$, and $M_i' = B_i' \cap T_i'$,
    for all $i$. Observe that since
    $\pl_\lambda \pr_\rho \notin \EGX{\RIGraph}$ for $\lambda > \rho$,
    $\cardin{M_i'} \leq 1$.  Let $B_i = B_i' \setminus M_i'$ and
    $T_i = T_i' \setminus M_i'$, for all $i$.  Consider the biclique
    cover \SoCGVer{%
       \begin{math}
           \BCB
           =
           \Set{ B_i  \otimes T_i,
              (B_i \cup T_i) \otimes M_i' }{ i \in \IRX{s} \bigr.}.
       \end{math}%
    }%
    \NotSoCGVer{%
       \begin{equation*}
           \BCB
           =
           \Set{ B_i  \otimes T_i,
              (B_i \cup T_i) \otimes M_i' }{ i \in \IRX{s} \bigr.}.
       \end{equation*}%
    }
    Here, if either side of a biclique is empty (probably $M_i'$)
    then it is not included in this cover. Clearly
    $\WeightX{\BCB} = \Theta(\WeightX{\BCA} =
    \Theta(\WeightX{\BC}$. Observe that for any
    $1 \leq i < j \leq n$ there exits a pair $X \otimes Y \in \BCB$
    such that $ij \in \EGX{\BCB}$.  Namely, $\BCB$ is a (disjoint)
    biclique cover of the clique $\binom{\IRX{n}}{2}$, and by
    \thmref{bi:c:k:n}  we have
    $\WeightX{\BC} = \Theta(\WeightX{\BCB}) = \Omega(n\log n)$.
\end{proof}

\section{Approximating the depth}

\newcommand{\RDepthX}[1]{\mathrm{d}^{\rectangleC}\pth{#1}}
\newcommand{\depthY}[2]{\mathrm{d}_{#1}\pth{#2}}
\newcommand{\maxDepthX}[1]{\mathrm{d}_{\max}\pth{#1}}
\newcommand{\maxRDepthX}[1]{\mathrm{d}^{\rectangleC}_{\max}\pth{#1\bigr.}}

For a set of regions $\Objs$ in the plane, the \emphi{depth} of a
point $\p\in\Re^2$, is the number of regions in $\Objs$ that contains
$\p$.  That is
$\depthY{\Objs}{\p} = \cardin{\Set{\region \in \Objs}{\p \in
      \region}}$.  The \emphi{maximum depth} of $\Objs$ is
$\maxDepthX{\Objs} = \max_{\p \in \Re^2} \depthY{\Objs}{\p}$.  For a
point set $\P$, consider the set of rectangles they induce
$\Rects = \RectsInfX{\P}$, see \Eqref{rect:set}.  For a point $\p$,
let $\RDepthX{\p}= \depthY{\Rects}{\p}$ denote the \emphi{rectangular
   depth} of $\p$.  Similarly, let
$\maxRDepthX{\P} = \maxDepthX{\Rects}$ denote the \emphi{maximum
   rectangular depth} of $\P$.

\begin{lemma}\RefProofInAppendix{depth:a}
    \lemlab{depth-approx}%
    Let $\P$ be a set of $n$ points in the plane, one can compute, in
    $O(n\log^4 n)$ time, a point $\p$ such that
    $\RDepthX{\p} = \Omega(\maxRDepthX{\P}/\log n)$.
\end{lemma}

\begin{proof:e}{\Xlemref{depth-approx}}{depth:a}
    We compute the implicit representation
    $\rectY{\P_1}{\PSA_1}, \ldots, \rectY{\P_s}{ \PSA_s}$ of
    $\Rects= \RectsInfX{\P}$, as described in \obsref{rects:implicit}.
    It is now straightforward to verify that for the set of rectangles
    $\rectY{\P_i}{\PSA_i}$, the depth along a horizontal line $\Line$
    can be described by a weighted set of
    $O( \cardin{\P_i} + \cardin{\PSA_i})$ interior disjoint intervals
    (that can also be computed in this time).

    Thus, computing the maximum depth on $\Rects$ on a vertical line
    $\Line$ can be reduced to computing the maximum depth of a point
    in a set of weighted intervals, where the total number of
    intervals is $\WeightX{\BC} = O( n\log^2 n)$. This in turn can be
    done in $O(n \log^3 n)$ time.

    We now apply the standard divide-and-conquer approach. Take the
    vertical line that is the median on the $x$-axis of the points of
    $\P$. Compute the maximum depth point on this line. Now, compute
    the maximum depth recursively on the point set on the right, and
    the point set on the left. Return the maximum depth point out of
    the three candidates computed. Observe that the recursion depth is
    $h = O( \log n)$, and this decomposes $\Rects$ into $O( \log n )$
    disjoint sets. One of these sets the maximum rectangular depth is
    $\maxRDepthX{\P}/h$, which is a lower bound on the depth of the
    point computed by the algorithm.
\end{proof:e}

\begin{lemma}
    \lemlab{eps:max:depth}%
    For any $\eps \in (0,1)$ and a point set $\P$ of $n$ points in
    the plane, one can construct, in $O( \eps^{-2} n \log^3 n )$
    time/space, a data structure that $(1-\eps)$-approximates
    rectangular depth queries in $\RectsInfX{\P}$. Specifically, given
    a query point $\q$, the data-structure returns, in $O(\log n)$
    time, a value $\alpha$ such that
    $(1-\eps) \RDepthX{\q} \leq \alpha \leq \RDepthX{\q} $.
\end{lemma}

\begin{proof}
    First compute the implicit decomposition
    $\rectY{\P_1}{\PSA_1}, \ldots, \rectY{\P_s}{ \PSA_s}$ of
    $\Rects= \RectsInfX{\P}$, using \thmref{smaller-biclique:2}.  The
    idea is to replace each set of rectangles $\rectY{\P_i}{\PSA_i}$,
    which is of size $\cardin{\P_i} \cdot \cardin{\PSA_i}$ by a set of
    interior disjoint weighted rectangles that approximates the depth
    they provide, and this new set has nearly linear size in
    $\cardin{\P_i} + \cardin{\PSA_i}$.

    So, consider such a pair $\rectY{\PA}{ \PB}$. By
    \obsref{rects:implicit} the point sets $\PA$ and $\PB$ are
    quadrant separated, and form each (say) a $\adom$-chain, so assume
    (without loss of generality) that $\PB \domby \PA$.  Let
    $\PA = \{ \pa_1, \ldots, \pa_s \}$ and
    $\PB = \{ \pb_1, \ldots, \pb_t\}$ be the two point sets sorted
    from left to right, and let
    $N = \cardin{\PA} + \cardin{\PB} = s + t$. Consider the
    arrangement $\Arr = \ArrX{\rectY{\PA}{\PB}\bigr.}$.

    Let $\BSC_k$ be the boundary of the (closure of the) set of all
    the points that dominate $k$ or more
    points of $\PB$. The set $\BSC_k$ is the \emphw{$k$\th bottom
       staircase} of $\Arr$. The \emphw{$k$\th top staircase} $\TSC_k$
    is defined analogously for $\PA$. If a point $\p \in \Re^2$ is
    between $\BSC_{k}$ and $\BSC_{k+1}$, and between $\TSC_{\ell}$ and
    $\TSC_{\ell+1}$, then $\depthY{\rectY{\PA}{\PB}}{\p} = k \ell$.
    See \figref{k:th:staircase:2}.

    \begin{figure}[h]
        \phantom{}%
        \hfill%
        \includegraphics[page=1]{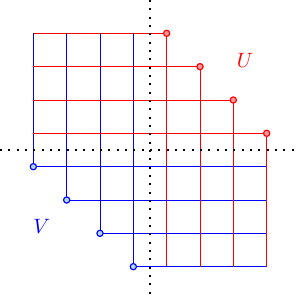}%
        \hfill%
        \includegraphics[page=2]{figs/stair_case_2}%
        \hfill%
        \phantom{}
        \caption{The structure of the staircases induced by the quadrant separated $\domby$-chains $\PA$ and
           $\PB$.}
        \figlab{k:th:staircase:2}
    \end{figure}

    Observe that the vertices of a top chain are quadrant-separated
    from the vertices of a bottom chain. Furthermore, given two bottom
    chains $\BSC_\alpha$ and $\BSC_\beta$, with $\alpha < \beta$, one
    can compute an axis aligned polygonal line $\gamma$ in between
    them of complexity $O\bigl( \cardin{\PB}/ ( \beta-\alpha)
    \bigr)$. Specifically, $\gamma$ starts between the top endpoints
    of $\BSC_\alpha$ and $\BSC_\beta$, and ends similarly between the
    right endpoints of the two staircases. Indeed, start with top
    vertex of $\BSC_\beta$ and move vertically down till hitting
    $\BSC_\alpha$. When this happens $\gamma$ changes direction and
    moves horizontally from left to right till hitting $\BSC_\beta$,
    then $\gamma$ changes back to moving vertically down. Keeping this
    alternation till exiting on the right.  To see the bound on the
    complexity of the curve consider the grid formed by horizontal and
    vertical lines passing through the points of $\PB$, and observe
    that the interior of every straight segment of $\gamma$ intersects
    $\beta-\alpha-1$ of these lines. Since $\gamma$ can intersect such
    a line only once, and there are $2 \cardin{\PB}$ such lines the
    bound follows.

    For $i=1, \ldots, \mu = \ceil{6/\eps}$, let $\alpha_i = i$.  For
    $i > \mu$, let
    $\alpha_i = \min( \cardin{\PB}, \ceil{ (1+\eps/3) \alpha_{i-1}}
    )$, and let $\tau$ be the first index such that
    $\alpha_i = \cardin{\PB}$.  It is easy to verify that
    $\tau = O( (1 + \log \cardin{\PB})/\eps )$.

    For $i=1,\ldots, \mu$ let $\BSC_i' = \BSC_i$, and for $i> \mu$,
    let $\BSC_i'$ be the above simplified curve lying between the
    $\alpha_i$\th and $\alpha_{i+1}$\th bottom staircase.  We repeat
    the same process for the top staircases. This results in a set of
    simplified staircases of total complexity $O(n /\eps)$, as the
    complexity of the simplified curves is dominated by the complexity
    of the first top/bottom $\mu$ curves. Importantly, every top curve
    intersects only $L = O( \eps^{-1} \log N)$ bottom curves and
    vice-versa. Thus, the total complexity of the arrangement of all
    these simplified top/bottom staircases is
    $N' = O( N/\eps + L^2 ) = O( N/\eps + (\log N)^2/\eps^2
    )$. Clearly, by sweeping, we can compute the approximate depth of
    each face of this arrangement, and this approximate depth is
    $(1+\eps)$-approximation for the depth of all the points inside
    this face. Every face can now be broken into rectangles (by this
    sweeping process). Thus yielding a set of interior disjoint
    (weighted) rectangles of size $N'$, such that the weight of the
    rectangle containing a point is the desired approximate depth in
    this biclique.  However, this arrangement is simply a grid in two
    of the quadrants where the top and bottom staircases intersect,
    and this portion can be computed directly without sweeping.  The
    other two quadrants are made out of only top or bottom staircases
    and both parts can be computed in linear time in their
    complexity. Namely, this rectangle decomposition can be computed
    in time proportional to its complexity.

    Repeating this process for all the bicliques, results in a set of
    rectangles of size
    \begin{equation*}
        O( \eps^{-1} n \log^2 n + \eps^{-2} n \log^2 n ),
    \end{equation*}
     since there are $O(n)$ bicliques and their total weight is
    $O( n \log^2 n)$.  This also bounds the total time to compute
    these rectangles. The final step is to overlay all these
    rectangles from all these bicliques and, using sweeping,
    preprocess them for point location and total depth query. Using
    persistence and segment trees for the $y$-structures increases the
    running time/space by a logarithmic factor, but now depth queries
    on this set of rectangles can be answered in $O( \log n)$ time,
    which implies the result.
\end{proof}

\begin{corollary}\RefProofInAppendix{max:depth}
    \corlab{max:depth}%
    Given $\eps \in (0,1)$, and a set $\P$ of $n$ points in the plane,
    one can approximate the maximum depth point in $\RectsInfX{\P}$ in
    $O( \eps^{-2} n \log^3 n)$ time.
\end{corollary}
\begin{proof:e}{\Xcorref{max:depth}}{max:depth}
    The above algorithm constructing the data-structure can also be
    used to compute the deepest point in the approximate set of
    rectangles, which readily yields the result.
\end{proof:e}

\section{On the structure of the box hull}
\seclab{box:hull}

For a point set $\P \subseteq \Re^2$, the \emphi{box hull} of $\P$ is
the region covered by the union of the rectangles in
$\Rects = \RectsInfX{\P}$, see \Eqref{rect:set}.  Formally, the box
hull of $\P$ is the set
\begin{math}
    \BHX{\P}%
    =%
    \cup \Rects
    =%
    \cup_{\rect \in \Rects} \rect.
\end{math}
\NotSoCGVer{%
   \begin{equation*}
    \BHX{\P}%
    =%
    \cup \Rects
    =%
    \cup_{\rect \in \Rects} \rect.
\end{equation*}
}%

The box hull seems initially deceptively simple, but unlike
convex-hull (and orthogonal convex-hull) it is not monotone, see
\figref{not:monotone}. It is not immediately obvious, for example,
that it does not have holes.

\begin{figure}[h]
    \phantom{}\hfill%
    \includegraphics[page=1]{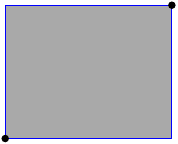}%
    \hfill%
    \includegraphics[page=2]{figs/not_monotone}%
    \hfill%
    \phantom{}%
    \caption{Adding points, might result in a smaller box hull.}
    \figlab{not:monotone}
\end{figure}

Let $\QuadX{\p}$ be the set of four closed axis parallel quadrants
centered at $\p$.

\begin{lemma}\RefProofInAppendix{c:rect}
    \lemlab{c:rect}%
    Consider a point $\p \in \Re^2$. If for all $T \in \QuadX{\p}$ we
    have that $T \cap \P \neq \emptyset$, then $\p \in \BHX{\P}$.
\end{lemma}

\begin{proof:e}{{\Xlemref{c:rect}}}{c:rect}
    Assume without loss of generality that $\p=(0,0)$ is the origin. Let
    $\p_1$ be the lowest point such that $\p_1 \domby \p$, i.e. the
    lowest point in the positive quadrant, let $\p_2$ be the lowest
    point such that $\p_2 \adom \p$, i.e. the lowest point in second
    quadrant (starting from the positive quadrant and enumerating
    counter-clockwise), and, similarly, let
    $\p_3$ and $\p_4$ be the highest points in the third and fourth
    quadrants. See \figref{axis:convexity} for an illustration.

    If the lower point of the pair $\p_1,\p_2$ and the higher point of
    the pair $\p_3,\p_4$ are in antipodal quadrants then we are done,
    as the rectangle they define contains the origin and is in
    $\Rects$. Otherwise, consider the horizontal slab between the higher
    point of $\p_1,\p_2$ and the lower point of $\p_3,\p_4$, without
    loss of generality these are $\p_2$ and $\p_3$, and denote the
    leftmost point in the first and fourth quadrant that is inside
    this slab by $\p'$, without loss of generality it is in the first
    quadrant. See \figref{axis:convexity}. We now get that
    $\rectY{\p'}{\p_3}$ is in $\Rects$, as the aforementioned slab
    does not contain any points of $\P$ in the left half of the plane
    due to the choice of $\p_2,\p_3$, and does not contain points of
    $\P$ with positive $x$-value that is lower than that of
    $\p'$. Since this rectangle contains the origin we are done.

    \begin{figure}
    	\centering \includegraphics{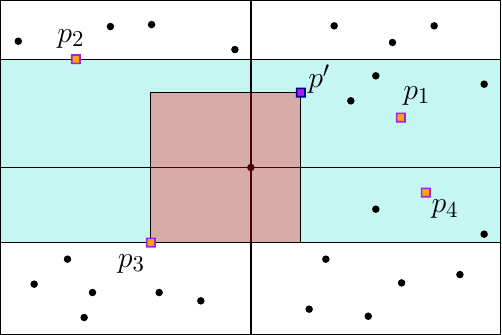}
    	\caption{An illustration of \lemref{c:rect}. The lowest
    		and highest points in the quadrants above and below the
    		$x$-axis respectively are marked by $p_i$. In this case the
    		pair of closest points of each side ($p_1,p_4$) is not
    		antipodal, and so the slab and the rightmost point in the slab
    		are also shown, as well as the resulted rectangle.}
    	\figlab{axis:convexity}
    \end{figure}

\end{proof:e}

\begin{lemma}\RefProofInAppendix{axis:convex}
    \lemlab{axis:convex}%
    The set $\BH = \BHX{\P}$ is connected. Furthermore, $\BH$ is
    horizontally and vertically convex -- formally, for any horizontal
    or vertical line $\Line$, either $\BHX{\P} \cap \Line= \emptyset$,
    or it is a segment.
\end{lemma}

\begin{proof:e:e}{{\Xlemref{axis:convex}}}{axis:convex}{\clearpage}
    We prove here that $\BHX{\P}$ is vertically convex, as the
    horizontal case is similar.  Let $\p,\q \in \BHX{\P}$ be two
    points on some vertical line. If $\p,\q$ are both contained in
    some rectangle of $\RIGX{\P}$ then we are done. Otherwise, we have
    two rectangles $\rectA, \rectB \in \RIGX{\P}$, such that
    $\p\in \rectA$ and $\q\in \rectB$, see \figref{up:down}.

    \begin{figure}[h!]
        \centering%
        \includegraphics{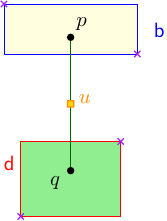}%
        \caption{}
        \figlab{up:down}
    \end{figure}

    Let $\pc$ be a point on the interior of the segment $\p\q$ that is
    outside $\rectA$ and $\rectB$. Without loss of generality, assume
    $\pc$ is the origin $(0,0)$. Since $\pc$ has $\rectA$ and $\rectB$
    above and below it, it must be that the four points defining them
    in $\P$ are in the four different quadrants of $\pc$.

    If one of the points defining $\rectA$ and $\rectB$ has an
    $x$-value of $0$ we can assume that this point is in the quadrant
    not containing its counterpart for creating the appropriate
    rectangle.  \lemref{c:rect} implies that $\pc$ is contained in a
    rectangle of $\RIGX{\P}$, and thus $\pc \in \BHX{\P}$. A
    contradiction.
\end{proof:e:e}

\begin{defn}
    The \emphi{$\domby$-maxima shadow}, denoted by
    $\shadowY{\domby}{\P}$, is the (closed) region in the plane of all
    points that dominates points in the $\domby$-maxima of $\P$. The
    \emphw{$\domby$-minima shadow} \emphw{$\adom$-maxima shadow}, and
    \emphw{$\domby$-minima shadow}, denoted respectively by
    $\shadowY{\Ndomby}{\P}$, $\shadowY{\adom}{\P}$, and
    ($\shadowY{\Nadom}{\P}$, are defined similarly.  See
    \figref{shadow}.
\end{defn}

Observe that the shadow regions can be computed readily from the
respective maxima/minima in linear time.

\begin{figure}[h]
    \centering%
    \includegraphics{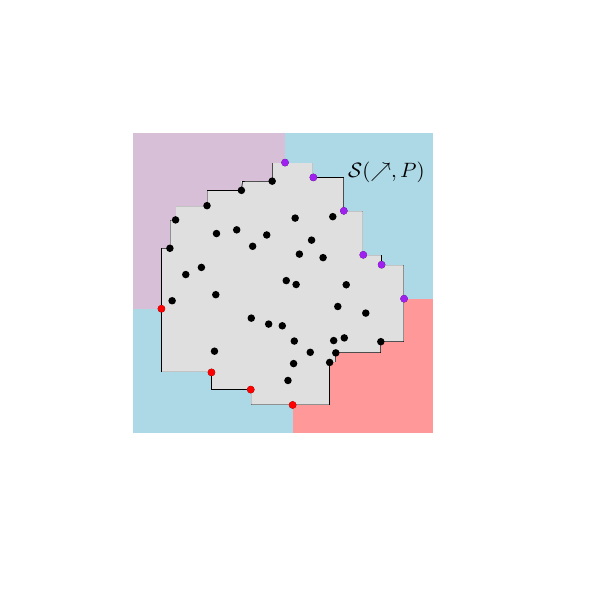}
    \caption{}
    \figlab{shadow}
\end{figure}

\begin{lemma}\RefProofInAppendix{b:h:chains}
    \lemlab{b:h:chains}%
    Let
    $\ShadowC= \shadowY{\domby}{\P} \cup \shadowY{\Ndomby}{\P} \cup
    \shadowY{\adom}{\P} \cup \shadowY{\Nadom}{\P}$. We have that
    $\BHX{\P} = \mathrm{closure}( \Re^2 \setminus \ShadowC)$.
\end{lemma}

\begin{proof:e}{\Xlemref{b:h:chains}}{b:h:chains}
    For simplicity of exposition we consider only points that are in
    the interior of the two sets. It is easy to verify that the
    argument can be modified to handle the boundary points.

    Consider a point $\q$ in the interior of
    $\shadowY{\domby}{\P} \subseteq \Re^2 \setminus C$, where
    $C = \mathrm{closure}( \Re^2 \setminus \ShadowC)$. The point $\q$
    dominates a point of $\p$ that belongs to the maxima of $\P$. It
    is straightforward to verify that any rectangle of
    $\Rects= \RectsInfX{\P}$ that contains $\q$, must also contain
    $\p$, which contradicts $\p$ being in the maxima. We conclude that
    $\q \in \Re^2 \setminus \BHX{\P}$.  Applying the same argument to
    the other three shadows, readily implies that
    $\Re^2 \setminus C \subseteq \Re^2 \setminus \BHX{\P}$, Namely,
    $\BHX{\P} \subseteq C$.

    Consider a point $\q \in \interX{C}$.  If the four quadrants of
    $\q$ are all non-empty of points of $\P$ then, by \lemref{c:rect},
    $\q \in \BHX{\P}$.  So assume that the top-right quadrant of $\q$
    is empty of points of $\P$.  This implies that $\q$ appears in the
    interior of the $\domby$-maxima of $\P \cup \{ \q \}$. In
    particular, there are two consecutive points
    $\p_1, \p_2 \in \maximaX{\P}$, such that
    $\rect = \rectY{\p_1}{\p_2} \in \Rects$, and $\q \in \rect$. This
    readily implies that $\q \in \BHX{\P}$.
\end{proof:e}

\begin{figure}
    \phantom{}\hfill%
    \includegraphics[page=6, width=0.40\linewidth]{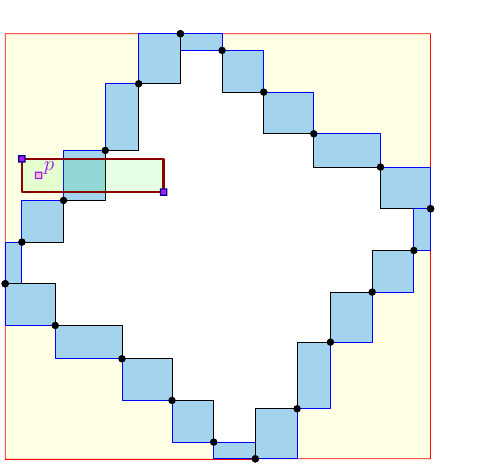}%
    \hfill%
    \includegraphics[page=7, width=0.40\linewidth]{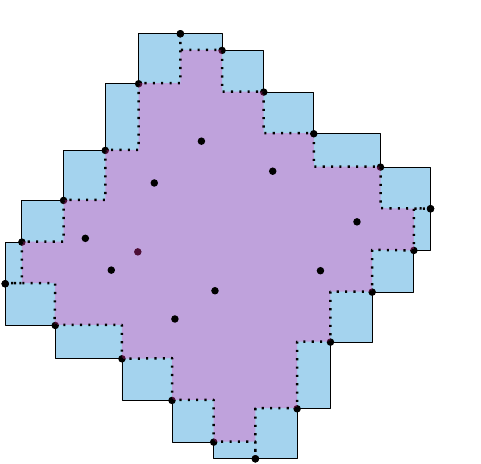}%
    \hfill%
    \phantom{}%
    \caption{The box hull of $\P$ Is defined by the four
       chains/anti-chains that constitute
       $\maximaX{\P}, \domMN{\P},\amaximaX{\P},$ and
       $\adomMN{\P}$. Left: If a point $\p$ of the union is outside of
       this perimeter a point of $\P$ must exist outside it as well,
       violating one of the chains. Right: A point set, its
       rectilinear convex hull (dotted line) and its box hull (full
       line).}
    \figlab{b:h:proof}
\end{figure}

\NotSoCGVer{%
   \subsection{Interior-disjoint cover of the box hull}%
}
\begin{lemma}\RefProofInAppendix{disjoint}
    \lemlab{interior:disjoint}%
    Let $\P$ be a set of $n$ points in the plane. One can compute, in
    $O(n \log n)$ time, a set $\Rects$ of $O(n)$ interior-disjoint
    rectangles of $\RIGX{\P}$, such that $\cup \Rects = \BHX{\P}$.
\end{lemma}

\begin{proof:e}{\Xlemref{interior:disjoint}}{disjoint}
    Consider the following algorithm for constructing a set $\Rects'$
    of disjoint rectangles given a set $\P$ of points. First, sort
    $\P$ by increasing $x$-values. Given the sorted set, consider the
    points according to their ordering while maintaining two sorted
    sets, $C$ and $C'$, that represent the $\domby$-maxima
    and $\adom$-minima of the points seen so far, i.e. the top right
    and bottom right extremal staircases.  Whenever a new point $\p$ is
    inserted we can check in $O(\log n)$ time the set of points
    $\P'=\{\p_1,...,\p_m\}$, sorted left to right, in $C$ or $C'$
    that, together with $\p$, define edges in $\RIGX{\P}$. Notice that
    $\p$ must be either below $C$ or above $C'$ depending on whether
    it is lower or higher than the last considered point, and thus one
    of $\P'\subseteq C$ or $\P'\subseteq C'$ must be true. We now add
    $\rectY{\p}{\p_1}$ to $\Rects$, and for every $i \in \{2,...,m\}$
    we add the rectangle
    $\rectY{\p}{\p_i} \setminus \rectY{\p}{\p_{i-1}}$. Notice that due
    to the choice of points, this set of rectangles is pairwise
    disjoint, and is also disjoint from every other rectangle add to
    $\Rects'$ thus far, as all previously added rectangles must reside
    below $C$, above $C'$, and between consecutive points of $C$ and
    $C'$. We now update $C$ and $C'$ using $\p$, i.e. replace
    $\p_1,...,\p_m$ with $\p$, or concatenate $\p$ to the end of the
    chain, depending on $\p$'s location with respect to the chain.

    We are now left with proving the claimed properties of $\Rects$
    other than disjointedness of interiors which has already been
    addressed. The size of $\Rects$ is straightforward to show as
    whenever a rectangle $\rectY{\p}{\p_i}$ is added we either have
    that $\p_i$ is the last point in $C$ and $C'$ (i.e. the last point
    introduced before $\p$), and after adding $\p$ we still have
    $\p_i$ on one of the chains, or $\p_i$ was only a member of one of
    $C,C'$, and was deleted when $\p$ was inserted. The first option
    can only happen once, when $\p_i$ is the last point in $C$ and
    $C'$, and the second option can also happen only once as it
    results in the removal of $\p_i$ from the data structure. This
    immediately implies an $O(n)$ size bound for $\Rects$.

    The runtime of the algorithm is also fairly simple, as an addition
    of a point to the chains requires a constant number of comparisons
    and binary searches in the sets in order to find $\p_1,...,\p_m$,
    after which the rectangles, which are computed in constant time
    each, are added. Since the number of rectangles is linear we get
    an $O(n\log n)$ runtime.

    Now, let $\p \in \bigcup\Rects$, without loss of generality
    $\p = (0,0)$, and let $\q,\pc$ be the rightmost points with
    negative $x$-value and positive and negative $y$-values respectively.
    Without loss of generality $x(\pc) < x(\q)$. Denote the horizontal
    line through $\q$ by $\ell$. If all of the points of $\P$ to the
    right of $\p$ are above $\ell$, then $\p\notin \bigcup\Rects$, so
    let $\pd$ be the leftmost such point. Now, depending on whether
    $y(\pd) > 0$ or $y(\pd) < 0$ we have that when $\pd$ was added to
    the data structure, a rectangle covering $\p$ was introduced to
    $\Rects$ due to $\pc$ or $\q$ respectively. This proves that
    $\BHX{\P} \subseteq \bigcup\Rects$, and since
    $\bigcup\Rects \subseteq \BHX{\P}$ is evident from the algorithm we get
    $\BHX{\P} = \bigcup\Rects$ as required.
\end{proof:e}

\begin{remark}
    Consider a point set $\P$, informally, the convex-hull of $\P$ can
    be defined by the repeated ``closure'' of $\P$ under the process
    of adding segments to it, that connect two points of $\P$ that do
    not have any point of $\P$ in its interior. Similarly, the box
    hull can be defined analogously -- the closure of the process of
    adding corners-induced ``empty'' rectangles to $\P$.
\end{remark}

\BibTexMode{%
   \bibliographystyle{plainurl}%
   \bibliography{rect_delaunay}%
}%
\BibLatexMode{\printbibliography}

\appendix%
\section{Chernoff's inequality}

\begin{theorem}[\cite{mr-ra-95}]
    \thmlab{chernoff}%
    Let $X_1,...,X_n$ be $n$ independent random variables such that
    $X_i\in[0,1]$. Denote $Y=\sum_{i=1}^n X_i$, and $\mu = \Ex{Y}$.
    Then, we have
    \begin{compactenumA}
        \smallskip%
        \item $\Prob{Y \geq (1+\delta)\mu} \leq \exp(-\delta^2\mu/4)$
        for $\delta \in (0,4)$.
        \smallskip%
        \item $\Prob{Y \leq (1-\delta)\mu} \leq \exp(-\delta^2\mu/2)$,
        for $\delta \in (0,1)$.
    \end{compactenumA}
\end{theorem}

\section{Additional results}

Here, we present some additional minor results.

\subsection{Random points}
\apndlab{random:points}

The following claim is well known \cite{m-cgitr-94}.
\begin{lemma}
    \lemlab{m:lower:bound}%
    Let $\P$ be a set of $n$ points picked uniformly, independently
    and randomly from $[0,1]^2$. Let $Y$ be the size of the maxima
    $\maximaX{\P}$. We have that $\Prob{ Y > 5 \ln n} \leq 1/n^2$ and
    $\Prob{Y < (\ln n)/16 } \leq 1/n^{2/5}$.
\end{lemma}

\begin{proof}
    Sorting the points from right to left in decreasing $x$-coordinate
    values, the $i$\th point has probability $1/i$ to be on the maxima
    of the first $i$ points, and these events are (somewhat
    surprisingly) independent \cite[Lemma 3.3]{hr-crwmv-15}.  Thus,
    $\mu = \Ex{Y} = \sum_{i=1}^n 1/i$. In particular
    $\ln n \leq \mu \leq \ln n + 1$. By \thmrefY{chernoff}{Chernoff's
       inequality}, \thmref{chernoff}, we have
    \begin{equation*}
        \Prob{Y > 5 \ln n }
        =
        \Prob{Y > (1+4)\mu}
        \leq
        \exp( - \mu 4^2/4)
        \leq
        1/n^4.
    \end{equation*}
    Thus, the variable $Y$ is bounded by $5\ln n$ with high
    probability. Similarly, for $\delta =15/16$, by
    \thmrefY{chernoff}{Chernoff's inequality}, we have
    \begin{align*}
      \Prob{Y \leq (\ln n)/16}
      \leq
      \Prob{Y \leq (1-\delta)\mu}
      \leq%
      \exp( -\delta^2 \mu/2)
      =
      \exp\Bigl( -\frac{225}{512} \mu\Bigr)
      \leq%
      1/n^{2/5}.%
      \SoCGVer{\tag*\qedhere}%
    \end{align*}
\end{proof}

\begin{lemma}
    \lemlab{h:p:maxima}%
    Let $\P$ be a set of points picked uniformly and independently
    from $[0,1]^2$. The number of edges in $\RIGraph = \RIGX{\P}$ is
    $\Theta(n\log n)$ with high probability $1 - O(1/n^{3/5})$.
\end{lemma}

\begin{proof}
    Let $\p \in \P$ be an arbitrary point of the set. we can use
    \lemref{m:lower:bound} four times, once for every quadrant
    defined by $\p$ in order to get an upper bound. Consider, say the
    points in the quadrant dominated by $\p$, i.e.
    $\P' = \P \cap \Set{\q\in [0,1]^2\cap \P}{\q \domby \p}$. Since
    $|\P'|\leq n$ we have that $|\maximaX{\P'}| > 5\ln n$ with
    probability at most $1/n^2$, and by using the claim for all four
    quadrants (with $\maximaX{\P}$, $\domMN{\P}$, $\amaximaX{\P}$ , or
    $\adomMN{\P}$ depending on the quadrant) we get that the
    probability that $\p$ is incident to more than $20 \ln n$ edges is
    at most $4/n^2$. Using the union bound immediately implies a
    probability of at least $1-4/n$ that every points is incident to
    $O(\ln n)$ edges in $\RIGX{\P}$.

    We now partition the unit square into nine squares
    $\square_{i,j}=[i/3,(i+1)/3]\times[j/3,(j+1)/3]$ for
    $i,j\in \{0,1,2\}$. The expected number of points, in each of
    these subsquares, is $\mu = n/9$, and
    \begin{equation*}
        \Prob{\cardin{\square_{i,j} \cap \P} < n/18} \leq \exp(-n/8),
    \end{equation*}
    by Chernoff's inequality.

    We therefore get that the probability that one of the nine squares
    contains less than $n/18$ points is $O(\exp(-n/8))$. Now, let
    $\p \in \P$ be a point in the center subsquare, i.e.
    $\p \in \square_{1,1} = [1/3,2/3]\times[1/3,2/3]$. Each of the
    quadrants defined by $\p$ fully contains a single subsquare
    $\square_{i,j}$, and therefore at least $n/18$ points. Denote
    $n' = n/18$. Since the four quadrants are disjoint, we have by
    \lemref{m:lower:bound} that the probability that the sizes of all
    the maxima and minima chains, of $\p$, is smaller than
    $\ln n' / 16$ is smaller than
    \begin{equation*}
        (1/n'^{2/5})^4%
        =%
        (1/n')^{8/5}%
        =%
        (18/n)^{8/5}.
    \end{equation*}
    The probability that any point, in the middle square, has less
    than $\ln n' / 16$ neighbors is therefore at most
    $O\bigl(\exp(-n/8)\bigr) +(n/18) (18/n)^{8/5} = O(1/n^{3/5})$ as
    required.
\end{proof}

\begin{remark}
    The high probability bound, for the lower bound of
    \lemref{h:p:maxima}, is somewhat weak. It can be straightened by
    breaking the square into a constant size subsquares, argue that
    the \RIG behavior inside each subsquare can be
    analyzed``independently'', and apply the bound of the lemma to
    each part. Thus showing that the number of edges of $\RIGraph$
    complexity is $\Omega(n \log n)$ with probability
    $\geq 1 - 1/n^{O(1)}$.
\end{remark}

\subsection{The structure of \TPDF{$\RIGX{\P}$}{RIG(P)}}
\apndlab{structure}

While $\RIGX{\P}$ can have a quadratic number of edges, it is fairly
straightforward to see that it does not have $K_5$ as a subgraph, as
Dilworth's Theorem guarantees the existence of a $\domby$ or
$\adom$-chain of size three for any set of five points in general
position, but if say $\p \domby \q \domby \pc$ we have
$\q \in \rectY{\p}{\pc}$ which means $\p\pc\notin \RIGX{\P}$. In this
section we show that not only does $\RIGX{\P}$ not contain large
cliques, it also does not contain large \emphi{tricliques},
i.e. disjoint triplets $R,G,B\subseteq \P$ such that every pair of
points from different sets are connected by an edge in
$\RIGX{\P}$. This means that the graph owes its size to the existence
of large bicliques. In the following we analyze the structure of
$\RIGX{\P}$ as discussed above, and give bounds on its size in the
case that $\P$ is a set of uniformly distributed random points.

\subsubsection{Bicliques and tricliques in \TPDF{$\RIGX{\P}$}{RIG(P)}}

Formally, a \emphi{triclique} over disjoint sets $X,Y,Z$ is the graph
$(X \otimes Y) \cup (X \otimes Z) \cup (Y \otimes Z)$.

\begin{lemma}
    \lemlab{biclique-structure}%
    Let $\P\subseteq \Re^2$ be a set of $n$ points. If $\RIGX{\P}$
    contains $K_{k,k}$, where $k > 4$, as a subgraph, then either
    $\RIGAX{\P}$ or $\RIGDX{\P}$ contain the biclique $K_{k,k-4}$ as a
    subgraph, formed by two or three chains, one of which is of length
    $k$.
\end{lemma}

\begin{proof}
    Denote the biggest biclique in $\RIGX{\P}$ by
    $C =B \otimes R$. We will refer
    to the points in $B$ and $R$ as the red and blue points. Let
    $C_b= (b_1,...,b_m)$ be the longest $\domby$-chain\slash
    anti-chain in $\G$. We assume without loss of generality that
    $C_b$ is a $\domby$-antichain of blue points, and notice that since $k>4$ we have $m>2$ (Dilworth's Theorem).
    The axis parallel lines through $b_1$ and $b_m$ partition
    $\Re^2$ into 9 parts $M_{1,1},...,M_{3,3}$, where $M_{1,1}$ is
    above and to the left of $b_1$, and $M_{3,3}$ is below and to the
    right of $b_k$. See \figref{nine:way:partition} for an
    illustration.

    $M_{1,1}$ and $M_{3,3}$ cannot contain any point of $C$ since blue
    points would break the maximality assumption on $C_b$, and red
    points would not be connected by an edge to both $b_1$ and $b_m$.

    Denote the downward generalized chain of $C_b$ by $C_b'$. Assume
    without loss of generality that at least $\frac{|R|}{2}$ of the
    red points are in the downward generalized quadrant of $C_b$,
    i.e. below $C_b'$ and in the union of $M_{2,1}, M_{3,1}$, and
    $M_{3,2}$, as $M_{3,3}$ does not contain any red or blue point,
    and let $C_r = (r_1,...,r_l)$ be the longest anti-chain of red
    points in that generalized quadrant.

    We now notice that any red point below $C_b'$ that is not a part
    of $C_r$ either prevents a point on $C_r$ from connecting to a
    point on $C_b$, or cannot connect to a point of $C_b$ because of a
    point in $C_r$. More precisely, for any two points $\p, \q$ below
    $C_b'$, if $\p \domby \q$, then, for any $1\leq i \leq m$, it is
    not possible that $\rectY{\p}{b_i}$ does not contain points of
    $\P$ in their interior. We therefore get that for every such pair
    of points either $\p \adom \q$, or $\q \adom \p$, and thus all of
    the red points bellow $C_b'$ are in $C_r$, meaning
    $|C_r|\geq \frac{k}{2}$, and since $|C_b|\geq |C_r|$ by
    assumption, we get that $|C_b|\geq \frac{k}{2}$.

    Also, no point of $R$ can anti-dominate or be
    anti-dominated by more than one point of $C_b$. Furthermore, for a
    point to be connected to all of $C_b$'s point by an edge in
    $\RIGX{\P}$, and anti-dominate or be anti-dominated by a point of
    $C_b$, that point must be in $M_{2,1}, M_{2,2},$ or $M_{3,2}$. See \figref{nine:way:partition}. If
    a red point $r$ exists in $M_{2,2}$, then it must be the case that
    $b_1 \adom r$, and $r \adom b_m$. In that case no other red point
    exists below $C_b'$, as no other point can connect to all of the
    points of $C_b$ with an edge. Since $k > 4$ and $|C_r|\geq \frac{|R|}{2}$ we have that $M_{2,2}$ does not have any points
    of $\P \setminus C_b$ below $C_b'$.

    $M_{2,1}$ and $M_{3,2}$ can only contain one red point each since
    there is no way to place two points in either cells that allows
    both points to connect to all of $C_b$ with an edge. If, say,
    $M_{2,3}$ contains a red point $r$, then it must be the case that
    $r\adom b_m$, and $b_i \domby r$ for every $1\leq i < m$. The case
    where a red point exists in $m_{3,2}$ is similar. We thus have
    that $|C_r|-2$ of the red points below $C_b'$ are in $M_{3,3}$,
    and these points create, together with $C_b$, $\adom$-chains
    forming a biclique in $\RIGDX{\P}$ of size at least as big as
    $K_{\frac{k}{2}, \frac{k}{2}-2}$.

    If there are no red points above $C_b'$, then we actually get that
    $|C_r| = k$, and thus $|C_b| \geq |C_r| = k$ as well, which
    implies a $K_{k, k-2}$ biclique in $\RIGDX{\P}$ formed by two
    $\adom$-chains of length $k$. Otherwise we notice that if a blue
    point $b$ is placed in $M_{2,1}$ or $M_{3,2}$, which can be done
    in a way that allows an edge between it and the single red point
    that might be in $M_{1,2}$ or $M_{2,3}$, $b$ can neither dominate
    or anti-dominate any point of $C_r$, as that would not allow the
    red point to connect to some points of $C_b$, and so it must be
    dominated or anti dominated by them. However, this is only
    possible if $|R|=k\leq2$, as $b$ can be connected by an edge to at
    most two points of the $\adom$-chain $C_r$ if it is ``between''
    them in $M_{3,1}$, or a single one if it is in $M_{2,1}$ or
    $M_{3,2}$. For similar reasons blue points can not be placed above
    $C_b$ unless $C_r$ is a single point in $M_{2,1}\cup
    M_{3,2}$. This means that $C_b = B$. Now we can use the same
    arguments used for $C_r$ to show that the red points above $C_b'$
    form an anti-chain, and that at most two of them can exist outside
    of $M_{1,3}$. In this case the three anti-chains (two red and one
    blue) form a $K_{k,k-4}$ biclique in $\RIGDX{\P}$, formed by 3
    $\adom$-chains, one of which is of length $k$.

	\begin{figure}
            \centering \includegraphics[page=4]{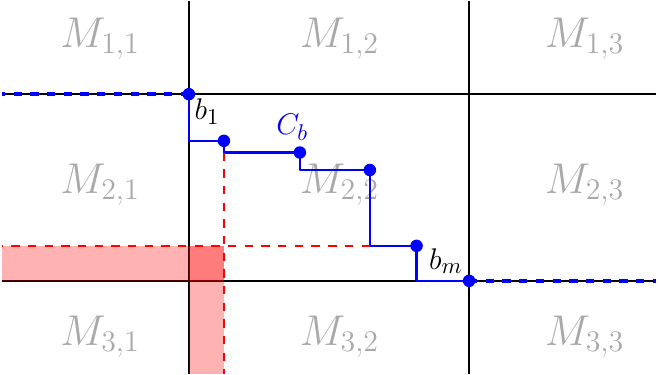}
            \caption{An illustration of the longest chain, and the
               partition it induces on $\Re^2$. The $\domby$-antichain
               $C_b$ is shown as a blue axis-parallel curve, the
               generalized $\domby$-antichain $C_b'$ is shown as the
               continuation of $C_b$ with dashed blue lines, and the
               region of $M_{2,1}\cup M_{2,2}\cup M_{3,2}$ that can
               contain red points is shown in red.}
            \figlab{nine:way:partition}
	\end{figure}
    \end{proof}

\begin{corollary}
    For a triclique $\{R,B,G\}$ in $\RIGX{\P}$ such that
    $|R|\leq|B|\leq|G|$, we must have $|G| = 1$.
\end{corollary}

\begin{proof}
    Since the proof is similar in flavor to that of
    \lemref{biclique-structure} in that it is a case analysis over the
    partition of the plane induced by the participating points, we use
    a simple proof by drawing using \figref{tricliques}. First we
    notice that the structure described in \lemref{biclique-structure}
    does not require that $k>4$ for anything but the size of a one
    directional biclique, and all of the arguments apply when
    $k\geq 3$. We therefore get that if $|B|\geq 3$ then the biclique
    $\{R,B\}$ includes, without loss of generality, an $\adom$-chain
    of three points of $R$, and an $\adom$-chain of two points of $B$.

    In the figure we have an $\adom$-chain of size three of $R$, and
    an $\adom$-chain of size two of $B$ which must exist due to
    arguments in \lemref{biclique-structure}. Notice that regardless
    of the position of the two blue points in the partition induced by
    the red points, there are no regions in which more than one point
    of $G$ can be added in order to form a triclique. The areas are
    coded to show the reason why a point of $G$ cannot be placed in
    that region; red if it cannot connect with an edge of $\RIGX{\P}$
    to all points of $R$, blue if it cannot connect to all points of
    $B$, purple if it prevents a point of $R$ and a point of $B$ from
    connecting, and a tiling pattern if it can only be connected to
    two or less blue points. Note that blue points cannot be placed in
    red regions and vice versa.
\end{proof}

\begin{figure}
    \phantom{}\hfill%
    \includegraphics[page=1,
    width=0.30\linewidth]{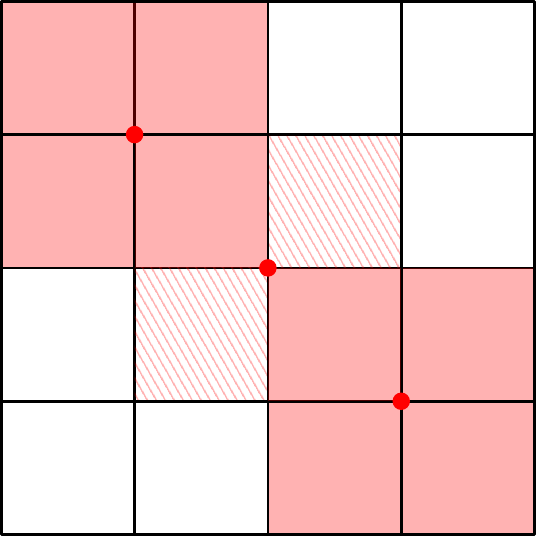}%
    \hfill%
    \includegraphics[page=2,
    width=0.30\linewidth]{figs/tricliques}%
    \hfill%
    \includegraphics[page=3,
    width=0.30\linewidth]{figs/tricliques}%
    \hfill%
    \phantom{}%
    \caption{A color and texture coded map of the proof for the
       size of tricliques in $\RIGX{\P}$}
    \figlab{tricliques}
\end{figure}

\subsection{Approximating independent set of rectangles}

\begin{lemma}
    Let $\Rects = \Set{\rectY{\p}{\q}}{\p\q\in\RIGX{\P}}$, and let $I$
    be
    \begin{equation*}
        \underset{X\subseteq \Rects}{\argmax}\Set{|X|}{\forall
           \rectA,\rectB\in X, \rectA\cap\rectB=\emptyset},
    \end{equation*}
    i.e. $I$ is a maximum size independent set of rectangles giving
    rise to edges of $\RIGX{\P}$. It is possible, in $O(n\log^2 n)$
    time, to find an independent set of $\Rects$ of size
    $\Omega(|I|/\log n)$.
\end{lemma}

\begin{proof}
    The proof here is similar to that of \lemref{depth-approx}, and we
    will therefore only describe the algorithm for finding a maximum
    size independent set of rectangles of $\RIGX{\P}$ intersecting a
    vertical line $\ell$. Since we can only take one rectangle from
    every biclique we traverse the biclique cover and keep, from every
    biclique that intersects $\ell$, only the rectangle created by the
    lowest point in $\P_i$ and the highest point in $\PSA_1$, assuming that $\P_i$ is above $\PSA_i$ (remember they are quadrant separated by \obsref{rects:implicit}). We then return the
    maximum independent set of intervals corresponding to the
    projections of these rectangle on $\ell$. The number of rectangles
    is $O(n)$, and thus we can get the maximum
    independent set in $O((n\log n)$ time using a greedy algorithm
    for maximum interval independent set.

    The rest of the proof, i.e. recursive algorithm and runtime
    analysis is similar to that of \lemref{depth-approx}, with the
    small differences of considering the sum of sets in each level of
    the recursion when keeping track of the maximum value, and the
    runtime itself which is dominated by the construction of the cover
    rather than the one-dimensional algorithm and thus results in an
    $O(n \log^3 n)$ overall runtime.
\end{proof}

\subsection{Induced subgraphs and edges}

\begin{lemma}
    For a set $\P\subseteq \Re^2$ of $n$ points in the plane, if $\P$
    does not contain a $\domby$-chain or anti-chain of length $l$,
    then $\RIGX{\P}$ contains $K_{k,k}$ as an induced subgraph for
    some $k=O(n/l)$.
\end{lemma}

\begin{proof}
    We first partition $\Re^2$ into four quadrants, such that two
    antipodal quadrants contain at least $n/4$ points of $\P$
    each. This can be done by finding a horizontal line that separates
    $\P$ into two sets of equal size ($\pm 1$), and then finding a
    vertical line with minimal $x$-coordinate that contains $n/4$
    points of one of the halves to its left.

    We now consider the two quadrants, without loss of generality
    these are the positive and negative quadrants (i.e. $(+,+)$ and $(-,-)$
    quadrants respectively) denoted $\Quad_1$ and $\Quad_3$, and denote
    $\P^+ = \Quad_1 \cap \P$ and $\P^- = \Quad_3 \cap \P$ . By Dilworth's
    Theorem, since the longest $\domby$-chain, in each quadrant is of
    length at most $l$, there is a partitioning of $\P^+$ and $\P^-$
    number into $l$ mutually disjoint $\domby$-antichains. The largest
    $\domby$-antichain in each of these sets is of length
    $\geq \frac{n}{4l}$, and the induced RIG over the set of points
    constituting these anti-chains is there $K_{k,k'}$ for
    $k,k' \geq \frac{n}{4l}$.
\end{proof}

\begin{lemma}
    Let $R$ be the set of all rectangles defined by points of $\P$
    (regardless of the number of points of $\P$ in their interior).
    There exists a point $\p\in\Re^2$ that is contained in $O(n/2)$
    rectangles of $R$ such that no two rectangles are supported
    by the same point, i.e. the set of points supporting the
    rectangles is $\P$. Furthermore we can find such a point in $O(n)$
    time.
\end{lemma}

\begin{proof}
    We start by finding a vertical line $\ell_v$ that splits the
    points of $\P$ into two parts of equal size. Let $\P_0=\P$, and
    let $x^+$ and $x^-$ be two counters initialized to 0. In the
    $i$\th iteration we find a horizontal line $\ell^i_h$ that splits
    $\P_{i-1}$ into two parts of equal size, let $\P_{i}^+$ be the
    points of $\P_{i-1}$ above $\ell_h$ and to the right of $\ell_v$,
    and let $\P_{i}^-$ be the points of $\P_{i-1}$ below $\ell_h$ and
    to the left of $\ell_v$. if $\P_{i}^+ + x^+ = \P_{i}^- + x^-$ we
    return the intersection point of $\ell_v$ and $\ell^i_h$,
    otherwise without loss of generality that
    $\P_{i}^+ + x^+ > \P_{i}^- + x^-$. We set
    $x^- = x^- + |\P_{i}^-|$, and
    $\P_{i} = \P_{i-1}\setminus \Set{\p\in\P}{\p\text{ is below
       }\ell^i_h}$.

    The correctness follows the simple observation that the process
    will always terminate when every two antipodal quadrants created
    by the lines will contain the same number of points in each
    quadrant. The runtime is given by the recursive formula
    $T(n)=T(n/2) + O(n)$ using linear time median selection. The
    formula gives an overall linear runtime.
\end{proof}

\InsertAppendixOfProofs

\end{document}